\documentclass[11pt, letterpaper]{article}

\usepackage[margin=1in]{geometry}
\usepackage{charter}

\usepackage{amsfonts, amsmath, amssymb, amsthm}
\usepackage{bm}
\usepackage{verbatim}
\usepackage{color}
\usepackage{euscript}
\usepackage{graphicx}
\usepackage[usenames,dvipsnames]{xcolor}
\usepackage{subcaption}

\usepackage{xspace}
\usepackage{wrapfig}

\newtheorem{theorem}{Theorem}
\newtheorem{lemma}{Lemma}
\newtheorem{corollary}{Corollary}
\newtheorem{observation}{Observation}

\newtheorem{definition}{Definition}

\newtheorem{claim}{Claim}

\usepackage[ruled,vlined,linesnumbered]{algorithm2e}
\usepackage{todonotes}
\usepackage{multirow}
\usepackage{thmtools}
\usepackage{thm-restate}
\usepackage{bbm}
\usepackage{cleveref}
\usepackage{authblk}

\title{Online and Dynamic Algorithms for Geometric Set Cover and Hitting Set}

\author[1]{Arindam Khan\thanks{supported by IUSSTF virtual center on ``Polynomials as an Algorithmic Paradigm', Pratiksha Trust Young Investigator Award,  Google India Research Award, Google ExploreCS Award, and SERB Core Research Grant     (CRG/2022/001176) on ``Optimization under Intractability and Uncertainty''.}}
\author[1]{Aditya Lonkar}
\author[1]{Saladi Rahul}
\author[1]{Aditya Subramanian\thanks{supported in part by Kotak IISc AI-ML Centre (KIAC) PhD Fellowship.}}
\author[2]{Andreas Wiese}
\affil[1]{Indian Institute of Science, Bengaluru, India}
\affil[2]{Technical University of Munich, Germany}

\newcommand{\IR}{\mathbb{R}}

\newcommand{\mT}{\mathbb{T}}

\newcommand{\awr}[1]{}
\newcommand{\asr}[1]{}
\newcommand{\akr}[1]{}
\newcommand{\alr}[1]{}
\newcommand{\rsr}[1]{}

\newcommand{\ak}[1]{{#1}}
\newcommand{\al}[1]{{#1}}

\begin{document}

\maketitle

\begin{abstract}
Set cover and hitting set are fundamental problems in combinatorial
optimization which are well-studied in the offline, online, and dynamic
settings. We study the geometric versions of these problems and present
new online and dynamic algorithms for them. In the online version
of set cover (resp. hitting set), $m$ sets (resp.~$n$ points) are given
 and $n$ points (resp.~$m$ sets) arrive online, one-by-one. In the dynamic
versions, points (resp. sets) can arrive as well as depart. Our goal
is to maintain a set cover (resp. hitting set), minimizing the size
of the computed solution.

For online set cover for (axis-parallel) squares of arbitrary sizes,
we present a tight $O(\log n)$-competitive algorithm. In the same
setting for hitting set, we provide a tight $O(\log N)$-competitive
algorithm, assuming that all points have integral coordinates in $[0,N)^{2}$.
No online algorithm had been known for either of these settings, not
even for unit squares (apart from the known online algorithms for
arbitrary set systems).

For both dynamic set cover and hitting set with $d$-dimensional hyperrectangles,
we obtain $(\log m)^{O(d)}$-approximation algorithms with $(\log m)^{O(d)}$
worst-case update time. This partially answers an open question
posed by Chan et al. {[}SODA'22{]}. Previously, no dynamic algorithms
with polylogarithmic update time were known even in the setting of squares (for either of these problems). 
 Our main technical contributions are an \emph{extended quad-tree
}approach and a \emph{frequency reduction} technique that reduces
geometric set cover instances to instances of general set cover with
bounded frequency. 
\end{abstract}
\global\long\def\OPT{\mathsf{OPT}}%
\global\long\def\P{\mathcal{P}}%
\global\long\def\NP{\mathsf{NP}}%
\global\long\def\ALG{\mathsf{ALG}}%
\global\long\def\RR{\mathbb{R}}%
\global\long\def\R{\mathbb{R}}%
\global\long\def\N{\mathbb{N}}%
\global\long\def\Z{\mathbb{Z}}%
\global\long\def\X{\mathcal{X}}%
\global\long\def\C{\mathcal{C}}%
\global\long\def\conv{\mathrm{conv}}%
\global\long\def\Prob{\mathrm{Prob}}%
\global\long\def\F{\mathcal{F}}%
\global\long\def\E{\mathbb{E}}%
\global\long\def\MST{\mathsf{MST}}%
\global\long\def\DP{\mathsf{DP}}%
\global\long\def\YES{\mathrm{YES}}%
\global\long\def\NO{\mathrm{NO}}%
\global\long\def\SIZE{\mathrm{SIZE}}%
\global\long\def\FF{\mathrm{FF}}%
\global\long\def\true{\mathsf{true}}%
\global\long\def\false{\mathsf{false}}%
\global\long\def\S{\mathcal{S}}%
\global\long\def\A{\mathcal{A}}%

\section{Introduction}

Geometric set cover is a fundamental and well-studied problem in computational geometry
\cite{clarkson2005improved, chan2012weighted, varadarajan2010weighted, har2012weighted, mustafa2014settling}.
Here, we are given a universe $P$ of $n$ points in $\R^{d}$, and a family $\S$ of $m$ sets, where each set $S\in\S$ is a geometric object (we assume $S$ to be a {\em closed} set in $\R^{d}$ and $S$ {\em covers} all points in $P\cap S$),
e.g., a hyperrectangle.
Our goal is to select a collection
$\S'\subseteq\S$ of these sets that contain (i.e., cover) all elements in $P$, minimizing the cardinality of $\S'$ (see \Cref{fig:problem-def} for an illustration). The {\em frequency} $f$ of the set system $(P,\S)$ is defined as the maximum number of sets that contain an element in $P$.

In the offline setting of
some cases of geometric set cover,  
better approximation ratios are known 
than those for the general set cover, e.g., there is
a polynomial-time approximation scheme (PTAS) for (axis-parallel)
squares \cite{MustafaR09}.
However, much less is understood in the online and in the dynamic
variants of geometric set cover. In the online setting, the sets are given offline and the points arrive one-by-one, and for an uncovered point, we have to select a (covering) set in an immediate and irrevocable manner. To the best of our knowledge, even for 2-D unit squares, there is no known online algorithm with an asymptotically improved competitive ratio compared to the $O(\log n\log m)$-competitive 
algorithm for general online set cover~\cite{alon2003online,BuchbinderN09}.
In the dynamic case, the sets are again given offline and at each time step a point is inserted or deleted.
Here, we are interested in algorithms that update the
current solution quickly when the input changes.  In particular, it is desirable to have algorithms
whose update times are polylogarithmic. Unfortunately, hardly any
such algorithm is known for geometric set cover. 
Agarwal et al.~\cite{AgarwalCSXX20} initiated the study of dynamic geometric
set cover for intervals and 2-D unit squares and presented $(1+\varepsilon)$-
and $O(1)$-approximation algorithms with polylogarithmic update times,
respectively. 
To the best of our knowledge, for more general objects, e.g., rectangles,
three-dimensional cubes, or hyperrectangles in higher
dimensions, no such dynamic algorithms are known. Note that in dynamic
geometric set cover, the inserted points are represented by their
coordinates, which is more compact than for  general (dynamic) set cover
(where for each new point $p$ we are given a list of the sets that
contain $p$, and hence, already to read this input we might need $\Omega(f)$
time).

\begin{figure}[!htb]
\centering
\includegraphics[page=11, scale=0.8]{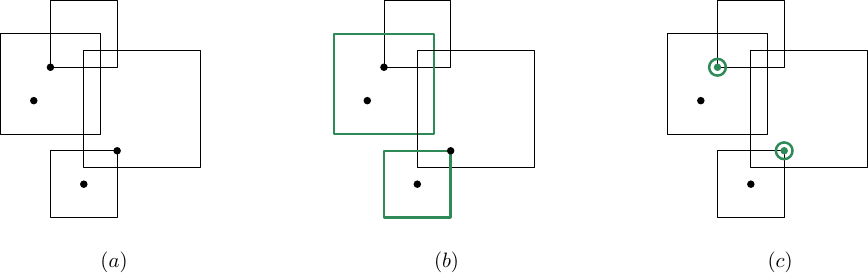}
\caption{(a) A set of squares $\S$ and a set of points $P$, (b) A set cover (in green) $\S'\subseteq \S$ covering $P$, (c) A hitting set (green points) $P'\subseteq P$ for $\S$.}
\label{fig:problem-def}
\end{figure}

Related to set cover is the hitting set problem (see \Cref{fig:problem-def} for an illustration) where, given a set of points $P$ and a collection of sets $\S$, we seek to select the minimum number of points $P' \subseteq P$ that hit each set $S\in \S$, i.e., such that $P' \cap S \ne \emptyset$. 
Again, in the offline geometric case, there are better approximation ratios known
than for the general case, e.g., a PTAS for squares~\cite{MustafaR09}, 
and an $O(\log\log\OPT)$-approximation  for rectangles~\cite{aronov2010small}.
However, in the online and the dynamic cases, only few results are
known that improve on the results for the general case. In the online
setting, there is an $O(\log n)$-competitive algorithm for $d=1$,
i.e., intervals, and an $O(\log n)$-competitive algorithm for unit
disks~\cite{EvenS11}. In the dynamic case, the only known algorithms
are for intervals and unit squares (and thus\al{,} also for quadrants),
yielding approximation ratios of $(1+\varepsilon)$ and $O(1)$, respectively~\cite{AgarwalCSXX20}.

\subsection{Our results}
In this paper, we study online algorithms for geometric set cover
and hitting set for squares \emph{of arbitrary sizes}, while previously no improved results were known even for unit squares.
Also, we present
dynamic algorithms for these problems for hyperrectangles of constant dimension $d$ (also called $d$-boxes or orthotopes) which are far more general geometric objects than those
which were previously studied, e.g., intervals \cite{AgarwalCSXX20} or (2-D)
squares \cite{chan2022dynamic}. 
\paragraph*{Online set cover for squares}
In \Cref{sec:Set-cover-squares} we study online set cover for axis-parallel
squares of arbitrary sizes and provide an online $O(\log n)$-competitive
algorithm. 
We also match (asymptotically) the lower bound of $\Omega(\log n)$,
and hence, our competitive ratio is tight. In our online model (as in \cite{alon2003online}), we
assume that the sets (squares) are given initially and the elements (points) arrive
online. 

Our online algorithm is based on a new offline algorithm that is \emph{monotone},
i.e., it has the property that if we add a new point $p$ to $P$,
the algorithm outputs a superset of the squares that it outputs given
only $P$ without $p$. The algorithm is based on a quad-tree decomposition.
It traverses the tree from the root to the leaves, and for each cell
$C$ in which points are still uncovered, it considers each edge $e$
of $C$ and selects the ``most useful'' squares containing $e$, i.e.,
the squares with the largest intersection with $C$. We assume (throughout
this paper) that all points and all vertices of the squares 
have integral coordinates in $[0,N)^{2}$ for a given $N$, and we
obtain a competitive ratio of $O(\log N)$. If we know that all the
inserted points come from an initially given set of $n$ candidate
points $P_{0}$ (as in, e.g., Alon et al.~\cite{alon2003online}), we improve our competitive ratio to $O(\log n)$. For this case,
we use the BBD-tree data structure due to Arya et al.~\cite{arya1998optimal}
which uses a more intricate decomposition into cells than a standard
quad-tree, and adapt our algorithm to it in a non-trivial manner.
Due to the monotonicity of our offline algorithm, we immediately obtain
an $O(\log n)$-competitive online algorithm.%

\paragraph*{Online hitting set for squares}
In \Cref{sec:hit-set-squares} we present an $O(\log N)$-competitive
algorithm for online hitting set for squares of arbitrary sizes, where the points are  given initially 
and the squares arrive online. This matches the best-known $O(\log N)$-competitive  algorithm for the much simpler case of intervals~\cite{EvenS11}. 
Also, there is a matching lower bound of $\Omega(\log N)$, even 
for intervals.

In a nutshell, if a new square $S$ is inserted by the adversary,
we identify $O(\log N)$ quad-tree cells for which $S$ contains one of its edges. Then, we pick the most useful points in these cells
to hit such squares: those are the points closest to the four edges
of the cell. We say that this \emph{activates} the cell. In
our analysis, we turn this around: we show that for each point $p\in\OPT$
there are only $O(\log N)$ cells that can possibly get activated
if a square $S$ is inserted that is hit by $p$. This yields a competitive
ratio of $O(\log N)$. 

\paragraph*{Dynamic set cover and hitting set for $d$-D hyperrectangles}
Then, in \Cref{sec:set-cover-hyperrectangles} and \ref{sec:hit-set-hyperrectangles}
we present our dynamic algorithms for set cover and hitting set for
hyperrectangles in $d$ dimensions. 
Note that no dynamic algorithm
with polylogarithmic update time and polylogarithmic approximation
ratio is known even for set cover for rectangles and it was asked
explicitly by Chan et al.~\cite{chan2022dynamic} whether such an
algorithm exists. Thus, we answer this question in the affirmative
for the case when only points are inserted and deleted. Note that
this is the relevant case when we seek to store our solution explicitly,
as discussed above. Even though our considered objects are very
general, our algorithms need only polylogarithmic worst-case update time. In contrast, Abboud et al.~\cite{abboud2019dynamic} showed that under Strong Exponential Time Hypothesis any general (dynamic) set cover algorithm with an amortized update time of $O(f^{1-\varepsilon})$ must have an approximation ratio of $\Omega(n^{\alpha})$ for some constant $\alpha>0$, and $f$ can be as large as $\Theta(m)$.

We first discuss our algorithm for set cover. We start with
reducing the case of hyperrectangles in $d$ dimensions to $2d$-dimensional hypercubes
with integral corners in $[0,4m]^{2d}$. Then, a natural approach would be to adapt our
algorithm for squares from above to $2d$-dimensional hypercubes.
A canonical generalization would be to build a quad-tree, traverse
it from the root to the leaves, and to select for each cell $C$ and
for each facet $F$ of $C$ the most useful hypercube $S$ containing
$F$, i.e., the hypercube $S$ with maximal intersection with $C$.
Unfortunately, this is no longer sufficient, not even in three dimensions:
it might be that there is a cell $C$ for which it is necessary that
we select cubes that contain only an edge of $C$ but not a facet
of $C$ (see \Cref{fig:CounterEx3D}). Here, we introduce a crucial
new idea: for each cell $C$ of the (standard) quad-tree and 
for each dimension $i\in [2d]$, consider the hypercubes which are ``edge-covering'' 
$C$ along dimension $i$. Based on these hypercubes a 
$(2d{-}1)$-dimensional recursive secondary structure is built on all the 
dimensions except the $i$-th dimension (see \Cref{fig:extended-quadtree}).

\begin{figure}[ht]
    \center
    \includegraphics[scale=0.3]{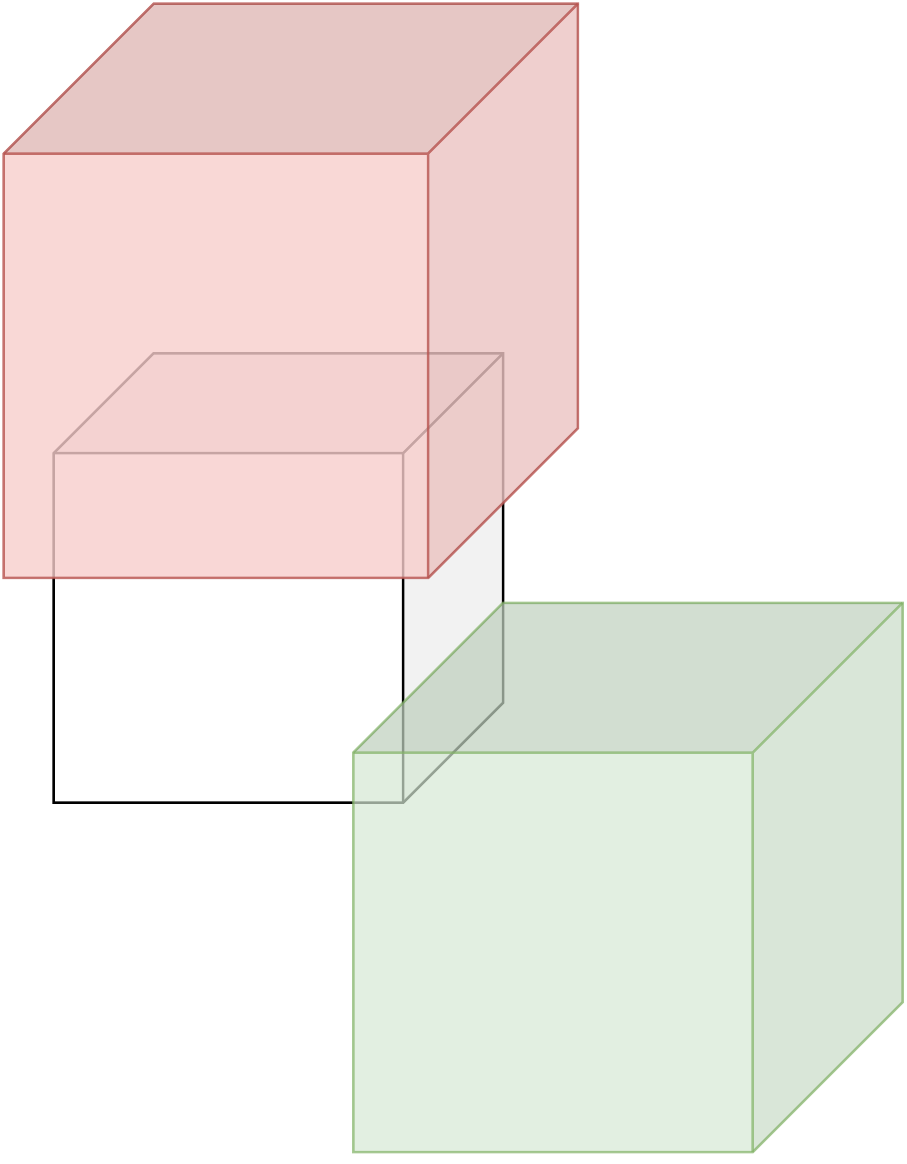}
    \caption{The red cube is the only cube that covers a facet of the (uncolored) cell. The green cube (from $\OPT$) only covers an edge of the cell. Note that there is no corner of a cube from $\OPT$ in the cell.
    Picking the red cube does not cover the
    the intersection of the green cube  with the cell. }
    \label{fig:CounterEx3D}
\end{figure}

We call the resulting tree the \emph{extended quad-tree. }Even though
it is much larger than the standard quad-tree, we show that each
point is contained in only $(\log m)^{O(d)}$ cells. Furthermore,
we use it for our second crucial idea to \emph{reduce the frequency}
of the set cover instance: we build an auxiliary instance of general
set cover with bounded frequency. It has the same points as the given
instance of geometric set cover, but different sets: for each node 
corresponding  to a one-dimensional cell $C$ of the extended quadtree
, 
we consider each of its endpoints $p$
and introduce a set that corresponds to the ``most useful'' hypercube
covering $p$,
i.e., the hypercube covering 
$p$
with maximal intersection
with $C$. Since each point is contained in only $(\log m)^{O(d)}$
cells, the resulting frequency is bounded by $(\log m)^{O(d)}$. Also,
we show that our auxiliary set
cover instance admits a solution with at most $\OPT\cdot(\log m)^{O(d)}$
sets. Then we use a dynamic algorithm from \cite{bhattacharya2021dynamic} for \emph{general }set
cover to maintain an approximate solution for our auxiliary instance,
which yields a dynamic $(\log m)^{O(d)}$-approximation algorithm.

We further adapt our dynamic set cover algorithm mentioned above to hitting set for $d$-dimensional hyperrectangles with an approximation ratio of $(\log n)^{O(d)}$.
Finally, we extend our algorithms for set cover and hitting set for
$d$-dimensional hyperrectangles even to the weighted case, at the
expense of only an extra factor of $(\log W)^{O(1)}$ in the update time and approximation ratio,
assuming that all sets/points in the input have weights in $[1,W]$.
See the following tables for a summary of our results. 

\begin{table}[!ht]
    \begin{centering}
        \begin{tabular}{|c|c|c|c|}
            \hline
            Problem & Objects & Competitive ratio & Lower bound \\
            \hline
            \hline
            \multirow{2}{*}{Set cover}  & intervals & 2 [Thm \ref{thm:onlineintervalupper}] & 2 [Thm \ref{thm:onlineintervallower}] \\ \cline{2-4}
                                        & 2-D squares & $O(\log n)$~[Thm \ref{squaressetcover_1}] & $\Omega(\log n)$~[Thm \ref{lb:unitsquares}] \\ \hline
            \multirow{2}{*}{Hitting set} & intervals & $O(\log N)$ \cite{EvenS11} & $\Omega(\log N)$ \cite{EvenS11} \\ \cline{2-4}
                                        & 2-D squares & $O(\log N)$~[Thm \ref{lem:squareshittingset_2}] & $\Omega(\log N)$\cite{EvenS11}  \\ \hline
        \end{tabular}
    \par\end{centering}
    \caption{Online algorithms for geometric set cover and hitting set.}
\end{table}

\begin{table}[!ht]
    \begin{centering}
        \begin{tabular}{|c|c|c|c|}
            \hline
            Problem & Objects & Approximation ratio & Update time \\
            \hline
            \hline
            \multirow{2}{*}{Set cover} 
                      & $2$-D unit squares & $O(1)$~\cite{AgarwalCSXX20} & $(\log n)^{O(1)}$ \\ \cline{2-4}
                      & $d$-D hyperrectangles & $O(\log^{4d-1}m)\log W$~[Thm \ref{thm:WtDynSetCov}] & $O(\log^{2d}m)\log^{3}(Wn)$ \\ \hline
            \multirow{2}{*}{Hitting set} 
                        & unit squares & $O(1)$~\cite{AgarwalCSXX20} & $(\log n)^{O(1)}$ \\\cline{2-4}  
                        & $d$-D hyperrectangles & $\ensuremath{O(\log^{4d-1}n)\log W}$~[Thm \ref{thm:WtDynHitSet}] & $O(\log^{2d-1}n)\log^{3}(Wm)$\\\hline
        \end{tabular}
    \par\end{centering}
    \caption{Dynamic algorithms for geometric set cover and hitting set. Update times in \cite{AgarwalCSXX20} are amortized and for the unweighted case. Our results are for worst-case update times. 
    }
\end{table}

\subsection{Other related work}
The general set cover is well-studied in both online and dynamic settings. 
Several variants and generalizations of online set cover have been considered, e.g., online submodular cover \cite{gupta2020online}, online set cover under random-order arrival \cite{gupta2022random}, online set cover with recourse \cite{gupta2017online}, etc. 

For dynamic setting, Gupta et al.~\cite{gupta2017online} initiated the study and provided $O(\log n)$-approximation algorithm with $O(f \log n)$-amortized update time, even in the weighted setting. Similar to our model, in their model sets are given offline and only elements can appear or depart. After this, there has been a series of works~\cite{abboud2019dynamic, bhattacharya2018deterministic, bhattacharya2019new, bhattacharya2018dynamic, bhattacharya2021dynamic, gupta2017online, gupta2020fully, assadi2021fully}.

Bhattacharya et al.~\cite{bhattacharya2021dynamic} have given deterministic $(1+\varepsilon)f$-approximation in \\ $O\left((f^2/\varepsilon^3) + (f/\varepsilon^2) \log (W)\right)$-amortized update time and $O(f\log^2(Wn)/\varepsilon^3)$-worst-case update time, where $W$ denotes the ratio of the weights of the highest and lowest weight sets. 
Assadi and Solomon \cite{assadi2021fully} have given a randomized $f$-approximation algorithm with $O(f^2)$-amortized update time. 

Agarwal et al.~\cite{AgarwalCSXX20} studied another dynamic setting for geometric set cover, where both points and sets can arrive or depart,  and presented $(1+\varepsilon)$-
and $O(1)$-approximation with sublinear update time for intervals and unit squares,
respectively.
Chan and He \cite{chan2021more} extended it to set cover with arbitrary
squares. Recently, Chan et al.~\cite{chan2022dynamic}
gave $(1+\varepsilon)$-approximation for the special case of intervals
in $O(\log^{3}n/\varepsilon^{3})$-amortized update time. They also gave $O(1)$-approximation
for dynamic set cover for unit squares, arbitrary squares, and weighted
intervals in amortized update time of $2^{O(\sqrt{\log n})},n^{1/2+\varepsilon}$,
and $2^{O(\sqrt{\log n\log\log n})}$, respectively.

Dynamic algorithms are also well-studied for other geometric problems such as maximum independent set of intervals and hyperrectangles \cite{Henzinger0W20, bhore2020dynamic, cardinal2021worst}, and geometric measure  \cite{dallant2021conditional}. 

\section{\label{sec:Set-cover-squares} Set cover for squares}

In this section we present our online and dynamic algorithms for set cover for squares. 
We are given a set of $m$ squares $\S$ 
such that each square $S\in\S$ has integral corners in $[0,N)^{2}$.  
 W.l.o.g.~assume that $N$ is a power of 2.
We first describe an offline $O(\log N)$-approximate algorithm. Then we 
construct an online algorithm and a dynamic algorithm based on it, such that both of them have
approximation ratios of $O(\log N)$ as well. For our offline algorithm, we assume that in addition
to $\S$ and $N$, we are given a set of points $P$ 
that we need to
cover, such that $P\subseteq[0,N)^{2}$ and each point $p\in P$ has integral coordinates.

\paragraph*{Quad-tree}
We start with the definition of a quad-tree $T=(V,E)$, similarly as in, e.g., \cite{arora1998polynomial, berg1997computational}. 
In $T$
each node $v\in V$ corresponds to a square cell $C_{v}\subseteq[0,N)^{2}$ whose vertices have integral coordinates. The
root $r\in V$ of $T$ corresponds to the cell $C_{r}:=[0,N)^{2}$. Recursively, consider a node $v\in
V$, corresponding to a cell $C_{v}$ and assume that
$C_{v}=[x_{1}^{(1)},x_{2}^{(1)})\times[x_{1}^{(2)},x_{2}^{(2)})$.  If $C_{v}$ is a unit square,
i.e., $|x_{2}^{(1)}-x_{1}^{(1)}|=|x_{2}^{(2)}-x_{1}^{(2)}|=1$, then we define that $v$ is a leaf.
Otherwise, we define that $v$ has four children $v_{1},v_{2},v_{3},v_{4}$ that correspond to the
four cells that we obtain if we {\em partition} $C_{v}$ into four equal sized smaller cells, i.e., define
$x_{\text{\ensuremath{\mathrm{mid}}}}^{(1)}:=(x_{2}^{(1)}-x_{1}^{(1)})/2$ and
$x_{\text{\ensuremath{\mathrm{mid}}}}^{(2)}:=(x_{2}^{(2)}-x_{1}^{(2)})/2$ and
$C_{v_{1}}=[x_{1}^{(1)},x_{\text{\ensuremath{\mathrm{mid}}}}^{(1)})\times[x_{1}^{(2)},x_{\text{\ensuremath{\mathrm{mid}}}}^{(2)})$,
$C_{v_{2}}=[x_{1}^{(1)},x_{\text{\ensuremath{\mathrm{mid}}}}^{(1)})\times[x_{\text{\ensuremath{\mathrm{mid}}}}^{(2)},x_{2}^{(2)})$,
$C_{v_{3}}=[x_{\text{\ensuremath{\mathrm{mid}}}}^{(1)},x_{2}^{(1)})\times[x_{1}^{(2)},x_{\text{\ensuremath{\mathrm{mid}}}}^{(2)})$,
and
$C_{v_{4}}=[x_{\text{\ensuremath{\mathrm{mid}}}}^{(1)},x_{2}^{(1)})\times[x_{\text{\ensuremath{\mathrm{mid}}}}^{(2)},x_{2}^{(2)})$.
Note that the
depth of this tree is $\log N$, where depth of a node in the tree is its distance from
the root of $T$, and depth of $T$ is the maximum depth of any node in $T$.
By the construction, each leaf node contains at most one point and it will lie on the bottom-left corner of the corresponding cell.

\paragraph*{Offline algorithm} 
In the offline algorithm $\mathcal {A}_\text{off}$, we
traverse $T$ in a breadth-first-order, i.e., we order the nodes in $V$ by their distances to the
root $r$ and consider them in this order (breaking ties arbitrarily but in a fixed manner). Suppose that in one
iteration we consider a node $v\in V$, corresponding to a cell $C_{v}$. We check whether the
squares selected in the ancestors of $v$ cover all points in $P\cap C_{v}$. If this is the case,
we do not select any squares from $\S$ in this iteration (corresponding to $v$). Observe that hence
we also do not select any squares in the iterations corresponding to the descendants of $v$ in
$T$ (so we might as well skip the whole subtree rooted at $v$).

Suppose now that the squares selected in the ancestors of $v$ do \emph{not }cover all points in
$P\cap C_{v}$. We call such a node to be {\em explored} by our algorithm.
Let $e$ be an edge of $C_{v}$. 
We say that a square containing $e$ is \emph{edge-covering for $e$}.
We select a square from $\S$ that is edge-covering for $e$ 
 and that has the largest intersection with
$C_{v}$ among all such squares in $\S$ (we call such a square {\em maximum area-covering}
for $C_v$ for edge $e$). We break ties in an arbitrary but fixed way, e.g., by
selecting the square with smallest index according to an arbitrary ordering of $\S$. If there is no
square in $\S$ that is edge-covering for $e$ then we do not select a square corresponding to $e$. We do this
 for each of the four edges of $C_{v}$. See \Cref{fig:offline-set-cov}. 
 If we reach a leaf node, and if there is an uncovered point (note that it must be on the bottom-left corner of the cell), then we select any arbitrary square that covers the point (the existence of such a square is guaranteed as some square in $\OPT$ covers it). See \Cref{fig:offline-edge}.


\begin{figure}
     \centering
     \includegraphics[page=10,width=\textwidth, scale=0.7]{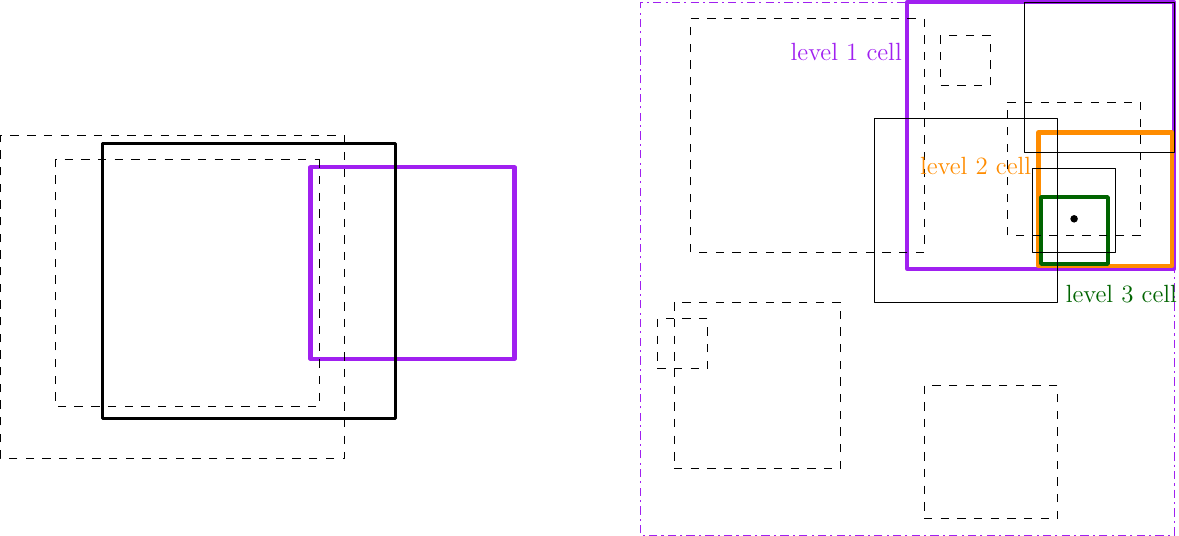}
    
    \caption{
    Left figure shows a quad-tree cell in purple. The maximum area-covering square (solid black) is picked, while the other edge-covering squares (dashed) are not.
Right figure shows the quad-tree cells (level-wise color-coded)
containing an uncovered point. In increasing order of depth of these cells, at most $4$ maximum-area covering squares (solid black) are picked together per cell, till the point gets covered.
    }
    \label{fig:offline-set-cov}
\end{figure}

\begin{figure}
\begin{center}
\includegraphics[page=14,scale=0.8]{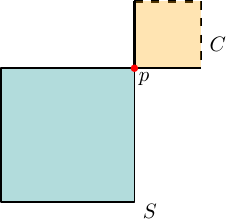}
\end{center}
\caption{Point $p$ lies in a leaf cell $C$ (which may not even have any edge-covering squares). In this case, we pick an arbitrary square $S$ to cover the point (since one such square always exists).}
\label{fig:offline-edge}
\end{figure}

\begin{restatable}{lemma}{offfeasible}
    \label{lem:offfeasible}
    $\mathcal{A}_\text{off}$ outputs a feasible set cover for the points in $P$.
\end{restatable}
    \begin{proof} Assume for contradiction that no square in ALG covered some point $p \in P$.  Since
    $\OPT$ is a feasible set cover, there is a square $S\in\OPT$ which covered $p$. There are two cases to consider here: either $p$ is exactly at one of the corners of $S$, or not. In the latter case, note that $S$
    is edge-covering for at least one  quad-tree cell containing $p$. Let $C_v$ be such a cell
    (which contains $p$ and its edge $e$ is contained in $S$) with minimum depth. Now the algorithm
    will traverse $T$ till we reach the node $v$ (corresponding to cell $C_v$)  containing $p$. As the
    squares selected by the algorithm for the ancestors of $v$ do not cover $p$, we will select the
    maximum area-covering square $S'$ (for $e$) in ALG. As $(S \cap C_v) \subseteq (S'\cap C_v)$, $S'$
    will cover $p$. This is a contradiction. Now in the first case, i.e., where $p$ is at one of the corners of $S\in\OPT$, either there is a leaf $v\in T$ which contains it and $S$ is edge-covering for $C_v$, or for such a leaf $v$, $S$ is corner-covering. In both the cases, $\A_{\text{off}}$ will pick a square for $v$ or one of its ancestors such that this square covers $p$.
\end{proof}


\paragraph*{Approximation ratio}
Let $\ALG\subseteq\S$ denote the selected set of squares and let $\OPT$ denote the optimal solution.
To prove the $O(\log N)$-approximation guarantee, the main idea is the following:
consider a node $v \in V$ and suppose that we selected at least one square in the iteration corresponding
to $v$. If $C_{v}$ contains a corner of a square $S\in\OPT$, then we charge the (at most four)
squares selected for $v$ to $S$. Otherwise, we argue that the squares selected for $v$ cover at
least as much of $C_{v}$ as the squares in $\OPT$, and that they cover all the remaining uncovered
points in $P\cap C_{v}$.  In particular, we do not select any further squares in the descendants of
$v$. The squares selected for $v$ are charged to the parent of $v$ (which contains a corner of a square $S \in \OPT$). Since each corner of each square
$S\in\OPT$ is contained in $O(\log N)$ cells, we show that each square $S\in\OPT$ receives
a total charge of $O(\log N)$.
Thus, we obtain the following lemma.
\begin{restatable}{lemma}{sqoffline}
    \label{lem:sqoffline}
    We have that $|\ALG|= O(\log N)\cdot|\OPT|$.
\end{restatable}
\begin{proof} We will charge each square picked in $\ALG$ to some square in $\OPT$. A cell $C_v$
    with its corresponding node $v$, can either contain (at least) a corner of some square in $\OPT$,
    or be edge-covered by (at least) a square in $\OPT$, or not intersect any square from $\OPT$ at
    all.
    \begin{itemize}
        \item $C_v$ contains a corner of $S\in\OPT$: In this case,  $\mathcal{A}_\text{off}$ picks at
            most four squares for the cell, and we charge these squares to a corner of $S$ in the cell. If
            there are multiple squares from $\OPT$ with a corner in the cell, pick one arbitrarily. This claim is true even when $C_v$ corresponds to a leaf node.
        \item Some square $S\in\OPT$ is edge-covering for $C_v$ (and $C_v$ has no corner of a square in $\OPT$): If
            $\mathcal{A}_\text{off}$ picks no edge-covering squares for such a cell, then we are fine.
            Otherwise, if $\mathcal{A}_\text{off}$ picks squares for such a cell, we claim that it covers
            all points in the cell
            . This is due to the fact that any point in this cell is covered by a
            square in $\OPT$ that is edge-covering for $C_v$, due to the absence of corners of squares of
            $\OPT$. So when $\mathcal{A}_\text{off}$ picks edge-covering squares with the largest
            intersection with the cell, the intersection of any square $S'\in\OPT$ with $C_v$ will also
            get covered). So, no child node of $v$ will be further explored by the algorithm. 
            This also means that the parent $v'$ of $v$ in the tree will contain a corner of $S$
            (because $C_{v'}$ intersects $S$, but cannot be edge-covered by it). We charge any squares
            picked by $\mathcal{A}_\text{off}$ at $C_v$ (at most four times) to this particular corner in
            the parent node. If there are multiple such corners, pick one arbitrarily.

        \item No squares from $\OPT$ intersect $C_v$: In this case, $C_v$ does not contain any points
            in $P$. Thus, $\mathcal{A}_\text{off}$ will not pick any squares for such a cell.
    \end{itemize}

    \begin{figure}[ht]
        \centering
        \includegraphics[page=4,scale=0.7]{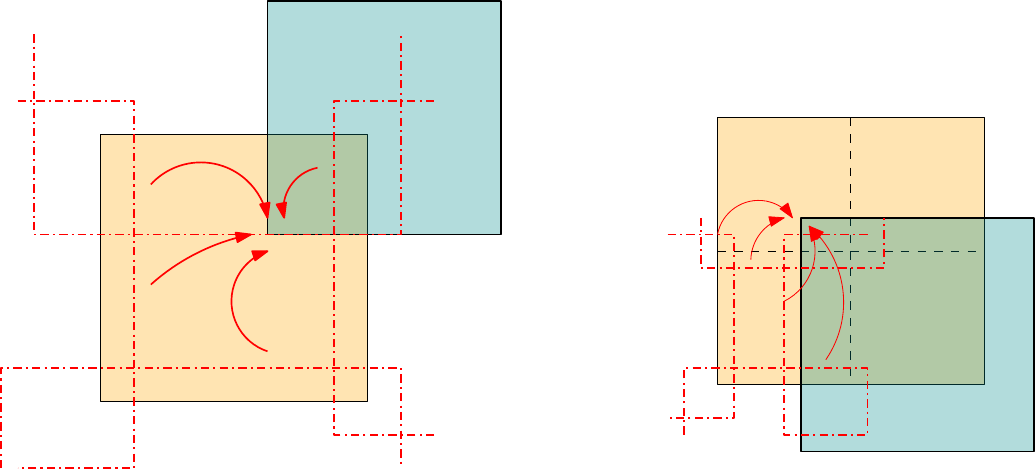}
        \caption{Charging picked (red) edge-covering squares to the corner of a (cyan) square in $\OPT$.
        In the image on the left, the (yellow) cell contains a corner of the square from $\OPT$, and in the image on the right, the parent of the cell contains such a corner.}
        \label{fig:2Dsetcov-charging}
    \end{figure}

    Now we note that a corner of any square in $\OPT$, will lie in at most $\log N$ cells of the
    quad-tree. For each of these cells, a corner is charged at most four times for the squares picked
    at the cell, and at most four times for each of its four child nodes. This amounts to a total
    charge of at most 20 per corner per cell. So each square in $\OPT$ is charged at most $20 \text{
    (per corner, per cell)}\times4\text{ (corners)}\times\log N\text{ (cells)}=80\log N$ times.
    Therefore, there are at most $80\log N\cdot|\OPT|$ squares in $\ALG$.
\end{proof}

\subsection{Online set cover for squares}

\subsubsection{$O(\log N)$-approximate online algorithm}

We want to turn our offline algorithm $\mathcal {A}_\text{off}$ into an online algorithm $\mathcal{A}_\text{on}$, assuming that in each {\em round} a new point is 
introduced by the adversary.
 The key insight for
this is that the algorithm above is \emph{monotone}, i.e., if we add a point to $P$, then it
outputs a superset of the squares from $\S$ that it had output before (when running it on $P$
only).  For a given set of points $P$, let $\ALG(P)\subseteq\S$ denote the set of squares that our
(offline) algorithm outputs.
\begin{restatable}{lemma}{monotone}
    \label{lem:monotone}
    Consider a set of points $P$ and a point $p$. Then $\ALG(P)\subseteq\ALG(P\cup\{p\})$.
\end{restatable}
\begin{proof} Assume towards contradiction that there exists some square $S$ in $\ALG(P)$ which did
    not belong to $\ALG(P\cup\{p\})$. According to the description of $\mathcal {A}_\text{off}$, we
    can infer that $S$ was picked by the algorithm in some iteration because it was maximum
    area-covering for some cell $C_v$ (corresponding to node $v$ in $T$) that contained a point $p' \in P$
    introduced by the adversary. Also, $\mathcal{A}_\text{off}$ in its run must have explored all the
    ancestors of $v$ in $T$.  Note that any such point $p'$ could be
    covered in a run of the algorithm only when it traverses cells that contain $p'$. This is due to
    the fact that
    once we pick some squares associated with a cell in the quad-tree, we only account for the area
    inside this cell that the squares cover. In light of this fact, 
    if $\mathcal{A}_\text{off}$ did not explore $C_v$ in this time
    step, then it also would not have explored the children of $v$ in $T$. Hence, the point $p'$
    would not have been covered which is a contradiction.
\end{proof}

Hence, it is easy now to derive an online algorithm for set cover for squares. Initially,
$P=\emptyset$. If a point $p$ is introduced by the adversary, then we compute $\ALG(P)$ (where $P$ denotes the set of previous points, i.e., \emph{without} $p$) and $\ALG(P\cup\{p\})$
and we add the squares in $\ALG(P\cup\{p\})\setminus\ALG(P)$ to our solution.
Therefore,
due to \Cref{lem:sqoffline} and \Cref{lem:monotone} we obtain
an $O(\log N)$-competitive online algorithm.

\subsubsection{$O(\log n)$-approximate online set cover for squares \label{sec:bbdscon}}

We assume now that we are given a set $\tilde{P}\subseteq\R^{2}$ with $|\tilde{P}|=n$ such that in
each round a point from $\tilde{P}$ is inserted to $P$,  
i.e., $P\subseteq\tilde
{P}$ after each round. 
We want to get a competitive ratio of $O(\log n)$ in this case. 
If $N=n^{O(1)}$ then this is immediate.  Otherwise, we extend our
algorithm such that it uses the balanced box-decomposition tree (or BBD-tree) data structure due to Arya et al.~\cite{arya1998optimal}, 
instead of the quad-tree. 
Before the first round, $P=\emptyset$ and we initialize the BBD-tree which 
yields a tree $\tilde{T}=(\tilde{V},\tilde{E})$ with the following properties:
\begin{itemize}
    \item each node $v\in\tilde{V}$ corresponds to a cell $\tilde{C}_{v}\subseteq[0,N)^{2}$ which is
        described by an outer box $b_{O}\subseteq[0,N)^{2}$ and an inner box $b_{I}\subseteq b_{O}$; both
        of them are axis-parallel rectangles and $\tilde{C}_{v}=b_{O}\setminus b_{I}$
        (Note that $b_I$ could be the empty set).
    \item the aspect ratio of $b_{O}$, i.e., the ratio between the length of the longest edge to the
        length of the shortest edge of $b_{O}$, is bounded by 3.
    \item if $b_{I}\ne\emptyset$, then $b_{I}$ is \emph{sticky }which intuitively means that in each
        dimension, the distance of $b_{I}$ to the boundary of $b_{O}$ is either 0 or at least the width of
        $b_{I}$. Formally, assume that $b_{O}=[x_{O}^{(1)},x_{O}^{(2)}]\times[y_{O}^{(1)},y_{O}^{
        (2)}]$ and $b_{I}=[x_{I}^{(1)},x_{I}^{(2)}]\times[y_{I}^{(1)},y_{I}^{(2)}]$. Then $x_{O}^{(1)}=x_
        {I}^{(1)}$ or $x_{I}^{(1)}-x_{O}^{(1)}\ge x_{I}^{(2)}-x_{I}^{(1)}$. Also $x_{O}^{(2)}=x_{I}^{
        (2)}$ or $x_{O}^{(2)}-x_{I}^{(2)}\ge x_{I}^{(2)}-x_{I}^{(1)}$. Analogous conditions also hold for
        the $y$-coordinates.
    \item each node $v\in\tilde{V}$ is a leaf or it has two children $v_{1},v_{2}\in\tilde{V}$; in the
        latter case $\tilde{C}_{v}=\tilde{C}_{v_{1}}\dot{\cup} \tilde{C}_{v_{2}}$.
    \item the depth of $\tilde{T}$ is $O(\log n)$ and each point $q\in[0,N)^{2}$ is contained in 
        $O(\log n)$ cells.
    \item each leaf node $v\in\tilde{V}$ contains at most one point in $\tilde{P}$.
\end{itemize}

In the construction of the BBD-tree, we make the cells at the same depth disjoint so that a point $p$ may be contained in exactly one cell at a certain depth. 
Hence, for a cell $\tilde{C}_{v}=b_{O}\setminus b_{I}$ we assume both $b_{O}$ and $b_{I}$ to be {\em closed} set, i.e., the boundary of the outer box $b_{O}$ is part of the cell and the boundary of the inner box $b_{I}$ is {\em not} part of the cell. 
We now describe an adjustment of our offline algorithm from \Cref{sec:Set-cover-squares}, working
with $\tilde{T}$ instead of $T$. Similarly, as before, we traverse $\tilde{T}$ in a
breadth-first-order. Suppose that in one iteration we consider a node $v\in\tilde{V}$ corresponding
to a cell $\tilde{C}_{v}$. We check whether the squares selected in the ancestors of $v$ cover all
points in $P\cap\tilde{C}_{v}$. If this is the case, we do not select any squares from $\S$ in this
iteration corresponding to $v$.

Suppose now that the squares selected in the ancestors of $v$ do
\emph{not }cover all points in $P\cap\tilde{C}_{v}$. 
Similar to \Cref{sec:Set-cover-squares}, we want to select $O(1)$ squares for $\tilde{C}_{v}$
such that if $\tilde{C}_{v}$ contains no corner of a square $S\in\OPT$, then the squares we selected for $\tilde{C}_{v}$ should cover all points in $P\cap\tilde{C}_{v}$. 
Similarly as before, for each edge $e$ of $b_{O}$ we select a square from $\S$
that contains $e$ and that has the largest intersection with $b_{O}$ among all such squares in
$\S$. We break ties in an arbitrary but fixed way.
However, as $\tilde{C}_{v}$ may not be a square and can have holes (due to $b_I$), apart from the edge-covering squares, we need to consider two additional types of squares in $\OPT$  with nonempty overlap with $\tilde{C}_{v}$:
(a) crossing $\tilde{C}_{v}$, i.e., squares that intersect two parallel edges of $b_O$; (b) has one or two corners inside $b_I$.

The following \textit{greedy subroutine} $\mathcal{G}$ will be useful in our algorithm to handle such problematic cases.
Let $R$ be a  box of width $w$ and height $h$ such that $w/h \le B$, for some constant $B \in \mathbb{N}$; and $P_R$ be a set of points inside $R$ that can be covered by a collection of vertically-crossing (i.e., they intersect both horizontal edges of $R$) squares $\S'$. 
Then, the set of squares picked according to $\mathcal{G}$ covers $P_R$ in the following way:
\begin{itemize}
    \item While there is an uncovered point $p'\in P_R$:
    \begin{itemize}
        \item Consider the leftmost such uncovered point $p\in P_R$.
        \item Select the vertically-crossing square intersecting $p$ (by assumption, such a square exists) with the rightmost edge.
    \end{itemize}
\end{itemize}
(The above subroutine is for finding vertically-crossing squares. For finding horizontally-crossing squares, we can appropriately rotate the input $90^{\circ}$ anti-clockwise, and apply the same subroutine.)
Then, we have the following claim about the aforementioned subroutine.
\begin{claim}
\label{cl:crossbbd}
Let $R$ be a  box of width $w$ and height $h$ such that $w/h \le B$, for some constant $B \in \mathbb{N}$; and $P_R$ be a set of points inside $R$ that can be covered by a collection of vertically-crossing (i.e., they intersect both horizontal edges of $R$) squares $\S'$. Then we can find at most $B+1$
squares from $\S'$ that can cover all points inside $R$. 

We have an analogous claim for horizontally-crossing squares when $h/w \le B$.
\end{claim}
\begin{proof}


Consider each iteration of the greedy subroutine $\mathcal{G}$. We call {\em pivot} to be  the leftmost point $p\in P_R$ that is not already
covered 
by a square selected by $\mathcal{G}$ so far. Then all selected vertically-crossing squares for $R$ 
will contain exactly one point that was identified as a pivot point at some point during the execution of the algorithm. As the aspect ratio is bounded by $B$ and the squares are vertically-crossing (i.e., their vertical length  is more than the vertical length of $R$), there can be at most $B+1$ pivot points. Hence, we select at most $B+1$ crossing squares  due to $R$. 
This produces a feasible set cover. 
\end{proof}

Now we describe our algorithm. 
First, we take care of the squares that can cross $b_O$. So, we apply the greedy subroutine $\mathcal{G}$ on $b_O$. As $b_O$ has bounded aspect ratio of 3, from Claim~\ref{cl:crossbbd}, we obtain at most $(3+1)+(1+1)=6$ squares that can cross $C_v$ vertically or horizontally. 
If $b_I=\emptyset$, we do not select any more squares. 
Otherwise, we need to take care of the squares that can have one or two corners inside $b_I$. 
Let $\ell_
{1},\ell_{2},\ell_{3},\ell_{4}$ denote the four lines that contain the four edges of $b_
{I}$. Observe that $\ell_{1},\ell_{2},\ell_{3},\ell_{4}$ partition $b_{O}$ into up to nine
rectangular regions, one being identical to $b_{I}$.
For each such rectangular region $R$, if it is sharing a horizontal edge with $b_I$, we again use $\mathcal{G}$ to select vertically-crossing squares. Otherwise, if $R$ is sharing a vertical edge with $b_I$, we use the subroutine $\mathcal{G}$ appropriately to select horizontally-crossing squares.
This takes care of squares having two corners inside $b_I$. 
Otherwise, if the rectangular region $R$ does not share an edge with $b_I$, then we check if there is a square  $S \in \S$ with a corner within $b_I$ that completely contains $R$.   We add $S$ to our solution too. This finally takes care of the case when a square has a single  corner inside $b_I$.  

Finally, to complete our algorithm, before its execution, we do the following: for every leaf $v$ for which $C_v$ contains at most one point $p\in \tilde{P}$, we associate a fixed square which covers $p$. Then, if our algorithm reaches a leaf $v$ while traversing that has an uncovered point $p$, we pick the associated square with this leaf that covers it. This condition in our algorithm guarantees feasibility.

\begin{figure}[ht]
    \centering
    \includegraphics[page=12,scale=0.6]{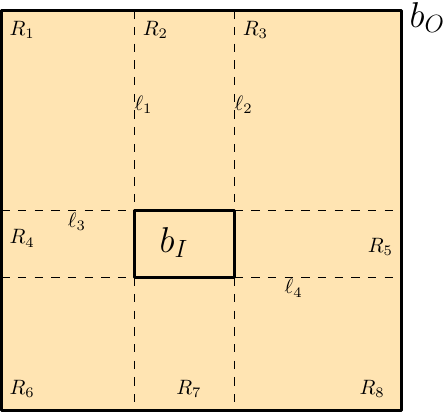}
    \caption{Outer box $b_O$ being partitioned into at most 9 rectangles due to inner box $b_I$.}
    \label{fig:bbd-boxes}
\end{figure}

\begin{figure}[ht]
    \centering
    \includegraphics[page=13,scale=0.6]{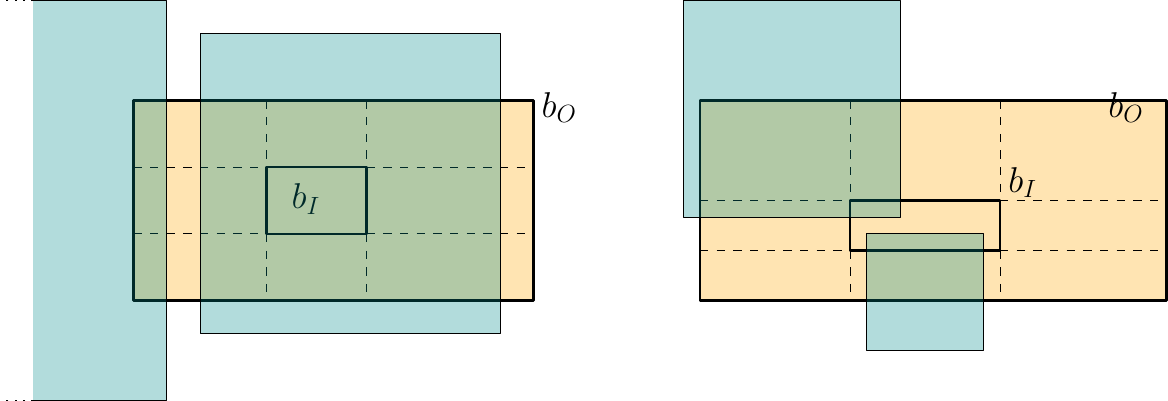}
    \caption{Possible intersections of a (cyan) square from $\OPT$ with a cell, such that no corner of the square is in the cell. The left image shows edge-covering, and crossing squares. The right image shows squares with one of two corners  inside $b_I$.}
    \label{fig:bbd-opt-intersection}
\end{figure}

\begin{restatable}{lemma}{bbdstruct}
\label{lem:bbdstruct}
Let $\tilde{C}_{v}$ be a cell such that the squares selected in the ancestors of $v$
    do \emph{not }cover all points in $P\cap\tilde{C}_{v}$. Then
    \begin{enumerate}
        \item[(a)] we select at most $O(1)$ squares for $\tilde{C}_{v}$ and
        \item[(b)] if $\tilde{C}_{v}$ contains no corner of a square $S\in\OPT$, then the squares we selected for
            $\tilde{C}_{v}$ cover all points in $P\cap\tilde{C}_{v}$.
    \end{enumerate}
\end{restatable}
\begin{proof}
First, we prove part (a). 
If $\tilde{C}_{v}$ corresponds to a  leaf node, we select at most one square. 
Otherwise, We select at most 4 edge-covering squares for $\tilde{C}_{v}$.
From Claim~\ref{cl:crossbbd}, we select $O(1)$ number squares for $\tilde{C}_{v}$ that are horizontally or vertically-crossing $b_O$. 
We select no more squares if  $b_{I}= \emptyset$.

So consider the other case:  $b_{I}\neq \emptyset$. Let $R$ be one of the (at most) four rectangular regions obtained from partitioning of $b_{O}$  (by $\ell_{1},\ell_{2},\ell_{3},\ell_{4}$) that share an edge with $b_I$.
Let $w,h$ be width and height of $R$, respectively. 
W.l.o.g.~assume $R$ shares a horizontal edge with $b_I$.
As  $b_{I}$ is sticky, and $b_{I}$ and $b_{O}$ have a bounded aspect ratio of 3, it can be seen that $R$ also has a $w/h \le 3$ (similarly, if $R$ shared a vertical edge with $b_I$, then $h/w \le 3$).
Again using Claim~\ref{cl:crossbbd}, we select $O(1)$ vertically-crossing squares for $R$.
We do a similar operation for other such regions. 
Now consider the remaining (at most four) regions obtained from partitioning of $b_{O}$  (by $\ell_{1},\ell_{2},\ell_{3},\ell_{4}$) that do not share an edge with $b_I$.
We select at most one square for each of them. 
Thus, in total, we select at most $O(1)$ squares for $\tilde{C}_{v}$. 

Now we prove part (b). 
If $\tilde{C}_{v}$ contains no corner of a square $S\in\OPT$, then all the squares in $\OPT$ that intersect $\tilde{C}_{v}$ are either edge-covering $b_O$, or crossing $\tilde{C}_{v}$, or contain one or two corners inside $b_I$. 
However, as we have picked maximal edge-covering squares for $b_O$, they contain all points in $\tilde{C}_{v}$ that are covered by the edge-covering squares from $\OPT$.

Similarly, by our selected  squares (that vertically or horizontally crosses $b_O$) in the greedy subroutine $\mathcal{G}$, we have covered all points that can be covered by such crossing squares in $\OPT$ that crosses $b_O$.

For squares that have (at least) a corner inside $b_I$, note that they have to cross one of the rectangular regions  that came from partitioning of $b_O$ and shares an edge with $b_I$. 
In fact, for such a square $S$ with a corner inside $b_I$, there is a rectangular region $R$ (say, with width $w$ and height $h$) among these four rectangular regions such that either $S$ is vertically-crossing for $R$ and $R$ shared a horizontal edge with $b_I$ (then $w/h\le 3$) or $S$ is horizontally-crossing for $R$ and $R$ shared a vertical edge with $b_I$ (then $h/w\le 3$). 
But then using Claim~\ref{cl:crossbbd}, we cover all points covered by such squares.
Additionally, a square that has exactly one corner inside $b_I$ may completely contain another rectangular region from partitioning of $b_O$ (that do not share an edge with $b_I$). For them, again we have covered them by selecting one square, if it exists.  
\end{proof}



Thus, we can establish a similar charging scheme as in \Cref{sec:Set-cover-squares}. To pay for our
solution, we charge each corner $q$ of a square $S\in\OPT$ at most $O(\log n)$ times. Hence, our
approximation ratio is $O(\log n)$. Similarly as in \Cref{sec:Set-cover-squares}, we can modify the
above offline algorithm to an online algorithm with an approximation
ratio of $O(\log n)$ each.


\begin{restatable}{theorem}{squaresetcoverone} \label{squaressetcover_1}
There is a deterministic $O(\log n)$-competitive online algorithm
for set cover for axis-parallel squares of arbitrary sizes.
\end{restatable}

\subsection{Lower bounds}

It is a natural question whether  algorithms having a competitive factor better than $O(\log n)$ are possible for online set cover for squares. We answer this question in the negative, even for the case of unit squares and even for randomized algorithms. We also remark here that there exists a tight $2$-competitive online algorithm for set cover for intervals (see Section~\ref{sec:interval_upper}).

\subsubsection{Unit squares and quadrants}

Given a set cover instance $(P, \F)$, for each $F\in\F$ define a variable $x_F$ which takes values
$\in[0,1]$. For a point $p\in P$, let $\F_p\subseteq\F$ be the sets that cover it. In the
{\em fractional set cover} problem, the aim is to assign values to the variables $x_F$ such that
for all points $p\in P$, $\sum_{F\in\F_p} x_F\ge1$.

\begin{figure}[ht]
    \center
    \includegraphics{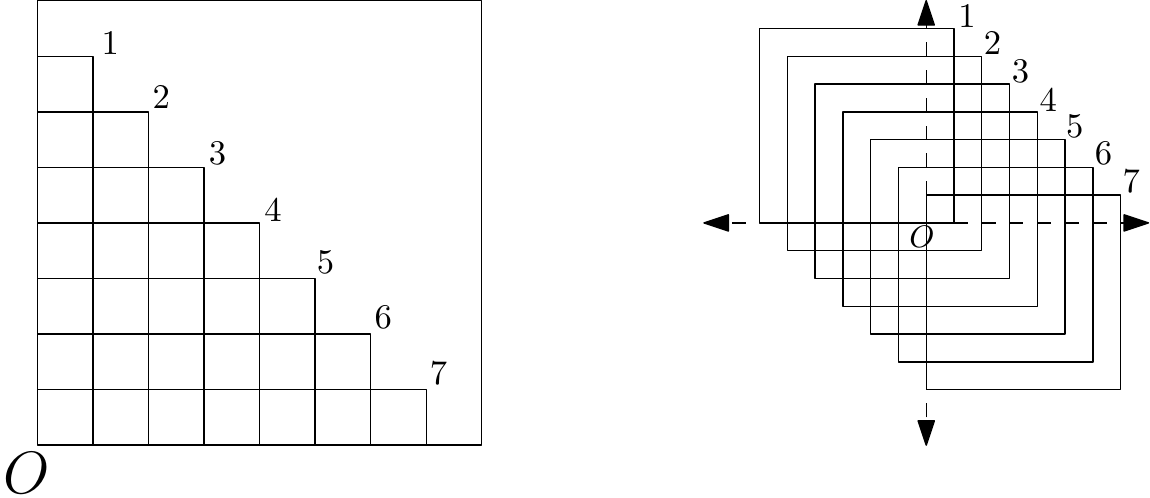}
    \caption{$\Omega(\log m)$ hard instances for quadrants and unit squares}
    \label{fig:LogHard}
\end{figure}

\label{subsec:sqLower}
\begin{restatable}{lemma}{quadlower} \label{lem:quad_lower}
    There is an instance of the online fractional covering problem on $m$ quadrants (quadrant is a  rectangle with one of its corners as the origin) such that any
    online deterministic algorithm is $\Omega(\log m)$-competitive on this instance.
\end{restatable}
\begin{proof} Consider the set of $m$ quadrants (See \Cref{fig:LogHard}) with their top-right
    corners as follows: $Q_1=(1, m), Q_2=(2, m-1), \dots, Q_k=(k, m-k+1), \dots, Q_m=(m, 1)$ (and let
    the corresponding variables for the set cover instance be $x_1, x_2, \dots, x_m$,  respectively).
    Now consider a point $p_{i,j}=(i-0.5, m-j+0.5)$. We claim that this point intersects exactly the
    $i^{th}$ to $j^{th}$ indexed quadrants. Since $p(x)>i-1$, $p_{i,j}$ does not intersect any of the
    first $(i-1)$ quadrants. Additionally since $p(y)=m-j+0.5<m-i-1$ it does intersect quadrant $i$
    and this holds true up to quadrant $j$ (Since, $p(y)=m-j+0.5<m-i-1$). Also $p_{i,j}$ does not
    intersect any of the last  $m-j+1$ quadrants because $p(y)>m-j$.

    Now consider an adversary that introduces points as follows: $P_1=p_{1,m}$. If the algorithm
    assigns values to the variables such that $\sum_{i=1}^{m/2}x_i\ge\sum_{i=m/2+1}^{i=m} x_i$, then the
    point $P_2=p_{m/2+1,m}$ is given, otherwise the point $P_2=p_{1,m/2}$ is given. The adversary
    repeats this process of halving the set of quadrants intersected, and puts the next point in
    the range with the lower sum, till only one quadrant, say quadrant $j$ is left.

    The optimal solution would have been to assign only $x_j$ to 1 and the remaining variables to
    zero. But any online algorithm can only halve the set of the potential optimal solution in each
    step, while assigning at least $1/2$ ``cost'' to the non-optimal quadrants. Hence, the cost of any
    online algorithm is $\Omega(\log m)$.
\end{proof}
\begin{corollary}
\label{lb:unitsquares}
    There is an instance of the online fractional covering problem on $m$ unit squares
    such that any deterministic online algorithm is $\Omega(\log m)$-competitive on this instance.
\end{corollary}
\begin{proof}
    We can appropriately extend the quadrants from \Cref{lem:quad_lower} in the bottom-left
    direction to obtain a similar instance on squares with side-length of $m$. More precisely,
    let the bottom left corner of the square $i$ corresponding to quadrant $i$ be $(i-m,1-i)$.
    The points introduced by the adversary are the same as in the quadrants instance.

    Now scale down this instance on squares appropriately, by a factor of $m$, to get the required
    unit square instance.
\end{proof}

Using standard techniques, as in \cite{bienkowski2020unbounded}, we can extend the lower bound for deterministic algorithms for the fractional variant to the lower bound for randomized algorithms for the integral variant. 

\begin{corollary}
    There is an instance of the online (integral) set cover problem on $m$ unit squares
    such that any randomized online algorithm is $\Omega(\log m)$-competitive on this instance.
\end{corollary}

Since in our lower bound construction, $n=\Theta(m^2)$, $\log n=\Theta(\log m)$ and hence, we have the theorem as stated below.

\begin{theorem}
    Any deterministic or randomized online algorithm for set cover for unit squares has a competitive ratio of $\Omega(\log n)$,
    even if all squares contain the origin and all points are contained in the same quadrant.
\end{theorem}

\section{Online hitting set for squares \label{sec:hit-set-squares}}

We present our online algorithm for hitting set for squares. We assume that we are given a fixed set
of points $P\subseteq[0,N)^{2}$ with integral coordinates. We maintain a set $P'$ of selected points
such that initially $P':=\emptyset$. In each round, we are given a square $S\subseteq[0,N)^{2}$ whose
corners have integral coordinates. 

We assume w.l.o.g.~that $N$ is a power of 2. Let $Q$ be all points with integral coordinates in
$[0,N)^{2}$, i.e., $P\subseteq Q$. 
For each point $q\in Q$ we say that
$q=(q_{x},q_{y})$ \emph{is of level $\ell$ }if both $q_{x}$ and $q_{y}$ are integral multiples of
$N/2^{\ell}$. We build the same quad-tree as in \Cref{sec:Set-cover-squares}.
We say that a cell $C_{v}$ \emph{is of level }$\ell$ if its height and width equal
$N/2^{\ell}$.

We present our algorithm now. Suppose that in some round a new square $S$ is given. If $S\cap
P'\ne\emptyset$ then we do not add any point to $P'$. Suppose now that $S\cap P'=\emptyset$. Let $q$
be a point of smallest level among all points in $Q\cap S$ (if there are many 
such points, then we select an arbitrary point in $Q\cap S$ of smallest level).
Intuitively, we interpret $q$ as if it were the origin and partition the plane into four {\em quadrants}. We define $O_{TR}:=\{(p_{x},p_{y})\mid p_{x}\ge q_{x}\, ,\, p_{y}\ge q_{y}\}$, and
$S_{TR}:=O_{TR}\cap S$, and define similarly $O_{TL},O_{BR},O_{BL}$,  and $S_{TL},S_{BR},S_{BL}$.
Consider $O_{TR}$ and $S_{TR}$. For each level $\ell=0,1,\dots,\log N$, we do the following. Consider
each cell $C$ of level $\ell$ in some fixed order  such that $C\subseteq O_{TR}$ and $S_{TR}$ is edge-covering
for some edge $e$ of $C$. Then, for each edge identify the point $p_b$
($p_t, p_l, p_r$, resp.) in $P\cap C$ that is closest to its bottom (top, left, and right, resp.)
edge. We add these (at most 4) points to our solution if at least one of $p_b,p_t,p_l,p_r$ is contained in $S_{TR}$ (see \Cref{fig:HitSet2D}).
If we add at least one such  point $p$ of the cell $C$ to $P'$ in this way, we say that $C$ gets
\emph{activated}. Note that we add possibly all of the points $p_b,p_t,p_l,p_r$ to $P'$ even though only one may be contained in $S_{TR}$. This is to ensure that $C$ gets activated at most once during a run of the online algorithm. This will be proved in Claim~\ref{claim:activate_once}, which will ultimately help us prove that our algorithm is $O(\log N)$-competitive. If for the current level $\ell$ we activate at least one cell $C$ of level
$\ell$, then we stop the loop and do not consider the other levels $\ell+1,\dots,\log N$. Otherwise,
we continue with level $\ell+1$.  We do a symmetric operation for the pairs ($O_{TL},S_{TL}$),
($O_{BR},S_{BR}$), and ($O_{BL},S_{BL}$).
We now prove  the correctness of the algorithm and that its competitive ratio is $O(\log N)$.

\begin{figure}[!ht]
    \center
    \includegraphics[scale=0.6, page=9]{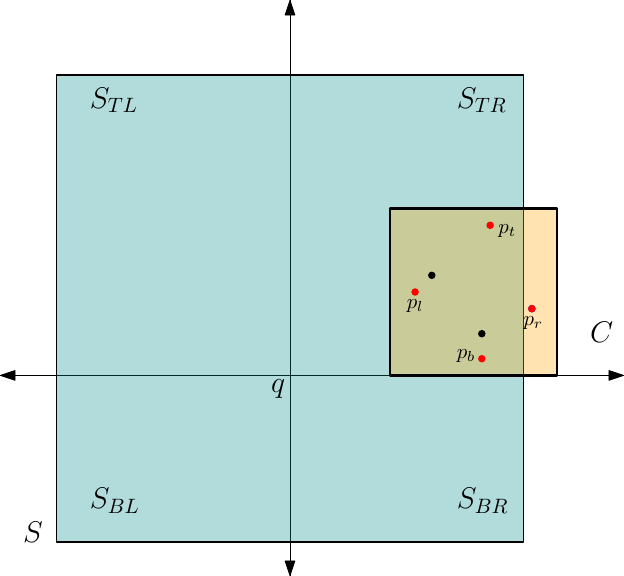}
    \caption{In cell $C$ lying in $O_\text{TR}$ the red points are chosen by the algorithm.}
    \label{fig:HitSet2D}
\end{figure}
\begin{restatable}{lemma}{sqhittingsetone}
\label{lem:sqhittingset_1}
After each round, the set $P'$ is a hitting set for the squares that have been added so far.
\end{restatable}
\begin{proof}
    We will prove that in case $S\cap P'=\emptyset$ when a square $S$ is inserted, at least one cell $C$ gets activated. Note that there exists a point $p$ in an optimum hitting set such that $p\in S$. Assume w.l.o.g.~that $p$ belongs to $S_{TR}$. Then, consider the set of cells $\C_p'$ that contain the point $p$. Since $S$ has side-length at least $1$ (it has integral coordinates for the corners) there exists a cell $C'\in\C_p'$ such that $S_{TR}$ covers an edge $e$ of $C'$. Hence, there will exist one level $\ell$ in $\{0,1,...,\log N\}$ such that a cell $C''\subseteq O_{TR}$  of level $\ell$ exists for which $S$ covered its edge (say, the bottom edge $e'$) and $p\in C''\cap S$. Then, our algorithm picked $p_b$ (point closest to the bottom edge) for $C''$, such that $p_b\in C''\cap S$. 
\end{proof}
Now we show that in each round $O(1)$
points are added to $P'$.
\begin{restatable}{lemma}{sqhittingsettwo} \label{lem:round-few-points}
\label{lem:sqhittingset_2}
In each round we add $O(1)$ points to $P'$.
\end{restatable}
\begin{proof}
We show that given a square $S$ such that $S\cap P'=\emptyset$ ($P'$ is the hitting set maintained by our algorithm), our algorithm activates at most $4$ cells. For this, we just observe that in each of the four quadrants $O_{TR}, O_{BR}, O_{TL}, O_{BL}$, we activate at most $1$ cell. For each of these cells, we pick at most $4$ points and hence, we add at most $16$ points in any round.
\end{proof}
Denote by $\OPT$ the optimal solution after the last round of inserting a square.
\begin{restatable}{lemma}{fewrounds} \label{lem:few-rounds}
    Let $p\in\OPT$. Then there are $O(\log N)$ rounds in which a square $S$ with $p\in S$ was
    inserted, such that at the beginning of the round  $P'\cap S=\emptyset$.
\end{restatable}
\begin{proof}
    First, define horizontal distance between two cells of level $\ell$ to be the distance between the 
    $x$-coordinates of their left edges. Analogously, define the vertical distance.  Now, we define a set of cells $\C_{p}$ corresponding to the point $p$,
    initialized to $\emptyset$. We will show later that in a certain round if a square $S$ introduced by the adversary contains $p$, and our algorithm activates at least one cell, then one cell in $\C_p$ is also activated. For each level $\ell\in\{0,1,\dots,\log N\}$, include in the set $\C_{p}$: 
    the cell $C$ of level $\ell$ containing $p$ and the other cells of level $\ell$ if they exist which have the same
    parent as $C$. We call these cells to be {\em primary} cells. Further, for level $\ell\in\{0,1,\dots,\log N\}$ consider the cell $C$ of level
    $\ell$ which contains $p$. Then, consider all the cells of level $\ell$ which are at a distance of
    $N/2^{\ell}$ horizontally but at a distance of $0$ vertically; also, consider cells which are at a distance of
    $N/2^{\ell}$ vertically but at a distance of $0$ horizontally from $C$.
    There can be at most $4$ such cells. For each such cell, include all its children in $\C_p$ (which
    are $4$ in number). We call these cells to be {\em secondary} cells. Hence, per level we select at most $4+4\times4=20$ cells. Therefore, $|\C_{p}|=O(\log N)$.

    \begin{figure}
        \centering
        \includegraphics[page=16,scale=0.75]{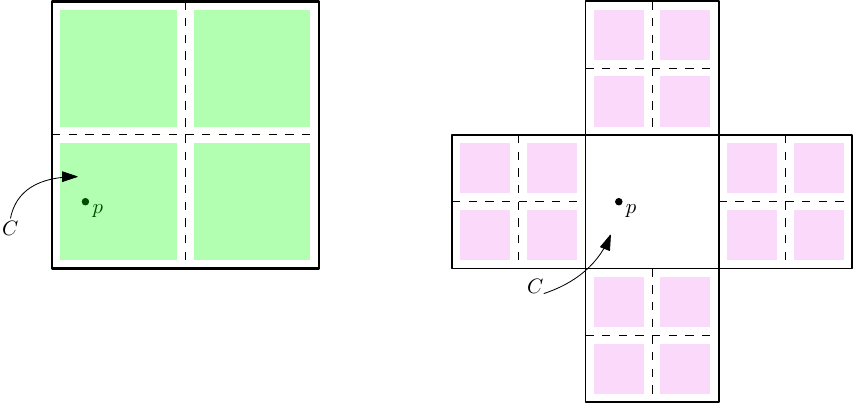}
        \caption{Cell $C$ contains point $p$. The primary cells are highlighted in green,
        while the secondary cells are highlighted in pink.}
        \label{fig:prim-sec-cells}
    \end{figure}

    We first observe that once a cell $\hat{C}$ gets activated, it does not get activated again.

    \begin{claim}
        A cell $\hat{C}\in T$ does not get activated more than once by our algorithm.
        \label{claim:activate_once}
    \end{claim}
    \begin{claim}
        Assume for contradiction that $\hat{C}$ was already activated in a previous round for a square $S'$.
        In the current round,   the square $S$ was introduced by the adversary such that $S\cap
        P'=\emptyset$. If $\hat{C}$ gets activated again, $S$ covered an edge of $\hat{C}$ (assume this is
        the bottom edge w.l.o.g.). If $\hat{C}$ was already activated in a previous round, then the
        algorithm must have picked points closest to the bottom, top, left, and right edges of $\hat{C}$ and
        at least one hit the square. Among these points, denote the point picked closest to the bottom edge
        by $p'$. If $\hat{C}$ was activated again for $S$, it clearly contained at least one point that hit
        $S$ by definition. Then, $p'$ would have hit $S$ since it had the lowest $y$-coordinate in
        $\hat{C}$.
    \end{claim}

    We now want to show that if a square $S$ is inserted in some round where $S\cap P'=\emptyset$ but $p\in
    S_{TR}$, then one cell in $\C_{p}$ gets activated
    but no cell $\hat{C}\notin\C_{p}$ with $\hat{C}\subseteq O_{TR}$ gets activated. 
    Once we prove this, observe
    that in each such round we activate one cell in $\C_p$. Then, by using the above claim that no cell in $\C_p$
    gets activated again in such a round, we are guaranteed that after $|\C_p|$ rounds, every square
    $S'$ introduced by the adversary which contained $p$ was hit by at least one point in the hitting
    set that the algorithm maintained.
    
    Then we do a symmetric argumentation for the cases that $p\in S_{TL}$, $p\in S_{BR}$, and $p\in S_
    {BL}$, each of them yielding the fact that if a square $S$ is added
    with $S\cap P'=\emptyset$ but $p\in S$, some cell among the cells in $\C_p$ gets activated. Thus, there can be only $|\C_p|=O(\log N)$ such rounds.
    Therefore, finally it remains to prove the claim.
    \begin{claim}
    \label{claim:few-rounds}
    If a square $S$ is inserted in some round where $S\cap P'=\emptyset$ but $p\in
    S_{TR}$, then one cell in $\C_{p}$ gets activated
    but no cell $\hat{C}\notin\C_{p}$ with $\hat{C}\subseteq O_{TR}$ gets activated.
    \end{claim}
    \begin{claim}
    Denote the level of $q$ which was one of the points at the smallest level  among points in $Q\cap S$ to be $\ell'$.
    Assume by contradiction that a cell $\hat{C}\notin\C_{p}$ with $\hat{C}\subseteq O_{TR}$ gets
    activated and $S$ w.l.o.g. covered its bottom edge. Let $\ell$ be the level of $\hat{C}$. Let $C$ be
    the cell of level $\ell$ containing $p$ and let $C'$ be its parent (which is at level $\ell-1$).  By the construction, we
    know that $\hat{C}\cap C'=\emptyset$ and hence, the parent of $\hat{C}$ is not $C'$. Then, the
    parent of $\hat{C}$ (denote by $C''$) and $C'$ are level $\ell-1$ cells at a distance  at least
    $N/2^{\ell-1}$, either horizontally or vertically. Assume w.l.o.g. that $C'$ is to the right side of
    $C''$. 
    By our assumption, right edge of $\hat{C}$ does not lie to the
    right side of the left edge of $C'$ (could coincide). 
    Any of the corners of $\hat{C}$ are points at level at most $\ell$. Then, we know
    that the level of $q$, which was $\ell'$ is at most $\ell$. 
    Now there are two cases.
    
    In the first case, the horizontal distance between $q$
    and the left edge of $C'$ is at least $N/2^{\ell-1}$ (see \Cref{fig:HittingSet-proof}(a)). Then $\ell'\leq \ell-1$ since then $S_{TR}$ covers
    the left bottom corner of $C'$. In this case, by our assumption $S_{TR}$ covers the bottom edge of
    $C''$ which also contains at least one point that hits it. This is a contradiction on the level of
    the activated cell in this round since $C''$ has level $\ell-1$. 

    In the other case,  the horizontal distance
    between $q$ and the left edge of $C'$ is strictly less than $N/2^{\ell-1}$ (see \Cref{fig:HittingSet-proof}(b)). In this case, $q$ is again at
    level exactly $\ell'=\ell-1$. Then, the right edge of $\hat{C}$ coincides with left edge of $C'$.
    Therefore, $C''$ is at a distance of exactly $N/2^{\ell-1}$ to the left of $C'$ and should have been added as a secondary cell. Hence,
    $\hat{C}\in\mathcal{C}_p$. This is a contradiction. 
    \end{claim}
    This completes the proof of the lemma. 
\end{proof}

\begin{figure}[ht]
 \centering
 \includegraphics[page=15,scale=.7]{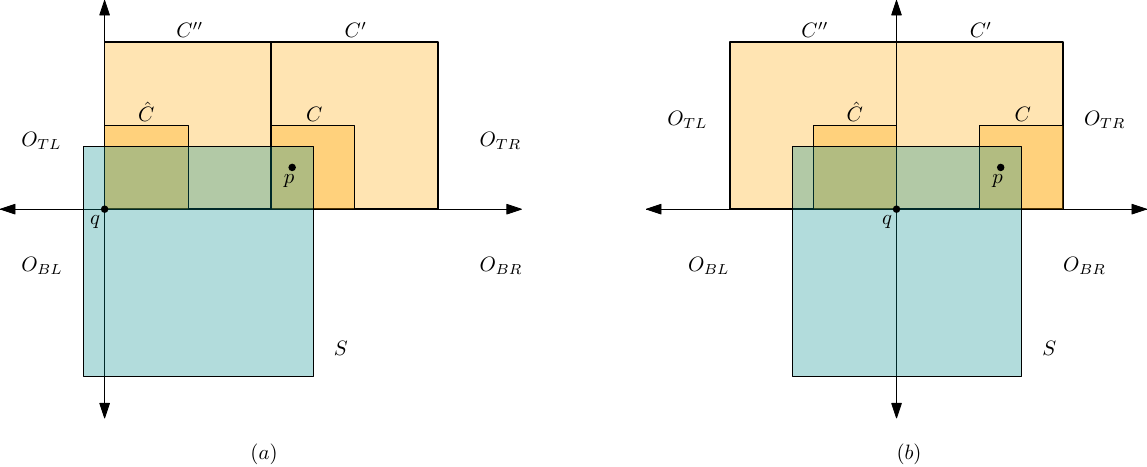}
 \caption{Proof of Claim~\ref{claim:few-rounds}. (a) Either, the distance between $C''$ and $C'$ is large, in which case $S$ is edge-covering for $C''$, or (b) $\hat{C}\in\C_p$.}
\label{fig:HittingSet-proof}
 \end{figure}


%
Hence, \Cref{lem:round-few-points} and \Cref{lem:few-rounds} imply that for each point $p\in\OPT$ we
add $O(\log N)$ points to $P'$. Thus, our competitive ratio is $O(\log N)$.
\begin{restatable}{theorem}{squarehittingsettwo}
\label{lem:squareshittingset_2}
    There is an $O(\log N)$-competitive deterministic online algorithm for hitting set for axis-parallel squares of arbitrary sizes. 
\end{restatable}

This is tight, as even for intervals, Even et al.~\cite{EvenS11} have shown an $\Omega(\log
N)$ lower bound.  

\section{Dynamic set cover for $d$-dimensional hyperrectangles}
\label{sec:set-cover-hyperrectangles} 

In this section, we will design an algorithm to dynamically maintain 
an approximate set cover for $d$-dimensional hyperrectangles.
 The main result we prove in 
this section is the following.

\begin{theorem}\label{thm:dyn-set-cover}
After performing a pre-processing step which takes 
$O(m\log^{2d}m)$ time, there is an algorithm for dynamic set cover for $d$-dimensional hyperrectangles 
with an approximation factor of $O(\log^{4d-1}m)$
 and an update time 
of $O(\log^{2d+2}m)$. 
\end{theorem}

Our goal is to adapt the quad-tree based algorithms designed 
in the previous sections of the paper. As a first step towards that, 
we  transform the problem such that 
the points and hyperrectangles in $\IR^d$ get transformed to points and {\em hypercubes}  
in $\IR^{2d}$, 
and the new problem is to cover the points in $\IR^{2d}$ with these hypercubes.
As discussed in the introduction, a simple $2d$-dimensional quad-tree on the 
hypercubes does not suffice for our purpose. We augment the 
quad-tree in two
ways: (a) at each node, we collect 
the hypercubes which are edge-covering w.r.t. that node and ``ignore'' 
that dimension in which they are edge-covering, and (b) recursively 
construct a $(2d{-}1)$-dimensional quad-tree on these hypercubes 
based on the remaining $2d{-}1$ dimensions. We call this new 
structure an {\em extended quad-tree}. The nice feature we obtain 
is that  any point in $\IR^{2d}$ will belong to 
only $O(\log^{2d}m)$ cells in the extended quad-tree. 
Furthermore, at the $1$-dimensional cells of the extended quad-tree, for each cell we
will 
identify
$O(1)$ ``most useful'' hypercubes.
This ensures that 
any point belongs to only $O(\log^{2d}m)$ most useful hypercubes. As a result, a ``bounded 
frequency'' set system can be constructed 
with the most useful hypercubes. The
 dynamic algorithm from Bhattacharya {\em et al.}~\cite{bhattacharya2021dynamic} (for general set cover) works efficiently 
 on bounded frequency set systems and 
 applying it in our setting leads to
 an $O(\log^{4d-1}m)$-approximation 
algorithm.

\subsection{Transformation to hypercubes in $\IR^{2d}$.} Recall that the input is a set $P$ of points and $\S$ is a collection 
of hyperrectangles in $\IR^d$. The first step of the 
algorithm is to transform the hyperrectangles in $\S$ to hypercubes in $\IR^{2d}$. 
Consider a hyperrectangle $S \in \S$ with $a=(a_1,\ldots,a_d)$ 
and 
$b=(b_1,\ldots, b_d)$ being the ``lower-left'' and the ``upper-right'' corners of $S$, 
respectively.
Let $\Delta=\max_{j=1}^d (b_j-a_j)$.
Then $S$ is transformed to a hypercube $S'$ in $\IR^{2d}$ with side-length 
$\Delta$ and ``top-right'' corner $(-a_1,-a_2,\ldots,-a_d,b_1,b_2,\ldots,b_d)$.
Let $\S'$ be the collection of these $m$ transformed hypercubes.
Let $P'$ be the set of $n$ points in $\IR^{2d}$ obtained by transforming 
each point $p=(p_1,\ldots,p_d)\in P$ to $p'=(-p_1,\ldots,-p_d,p_1,\ldots,p_d)$.

\begin{figure}[h]
 \begin{center}
      \includegraphics[scale=1]{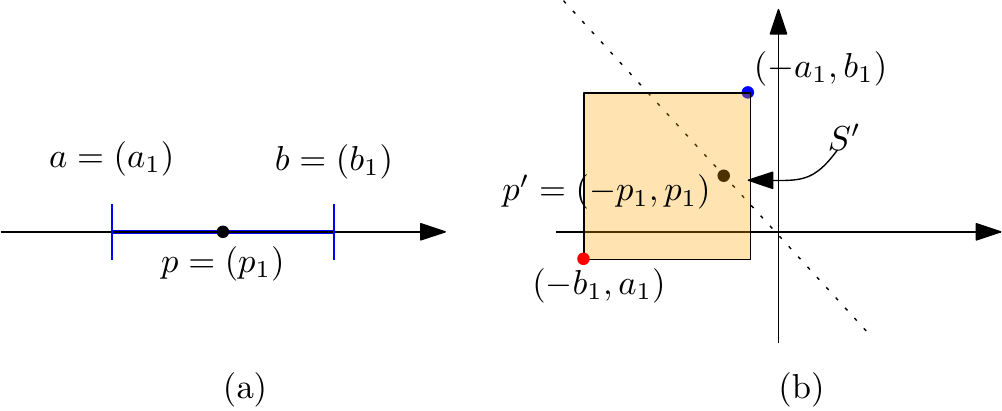}
 \end{center}
 \caption{(a) A point $p$ in 1-D lying inside an interval $S=[a_1,b_1]$, and 
 (b) the transformation of $p$ into a point $p'=(-p_1,p_1)$, 
 and the transformation of $S$ into a 
 square $S'$ in 2-D.}
 \label{fig:set-cover-dual}
\end{figure}

\begin{observation}
A point $p=(p_1,\ldots,p_d)$ lies inside $S$ if and only if 
$p'=(-p_1,\ldots,-p_d,p_1,\ldots,p_d)$ lies inside $S'$.
\end{observation}
\begin{proof}
Assume that $p$ lies inside $S$.
Consider the $i$-th coordinate of $p'$ with $i\leq d$. 
Since $b_i -a_i \leq \Delta$ implies that $ -a_i -\Delta \leq -b_i$, 
we observe that $-a_i - \Delta \leq -b_i \leq -p_i  \leq -a_i$.
Therefore, for all $1\leq i \leq d$, we have $-a_i-\Delta \leq -p_i \leq -a_i$.
 
Now consider the $i$-th coordinate of $p'$ with $i >d$. 
Since $b_i -a_i \leq \Delta$ implies that $ b_i -\Delta \leq a_i$, 
we observe that $b_i - \Delta \leq a_i \leq p_i \leq b_i$.
Therefore, for all $d+1\leq i \leq 2d$, we have $b_i-\Delta \leq p_i \leq b_i$.
Thus, we claim that if a point $p$ lies inside $S$, then $p'$ will lie inside 
$S'$. 

It is easy to prove the other direction. If $p$ lies outside $S$, then there is 
at least one coordinate (say $i$) in which either $p_i <a_i$ or $p_i >b_i$. 
If $p_i < a_i$, then $-p_i > -a_i$ and hence, $p'$ lies outside $S'$. 
On the other hand, if $p_i>b_i$, then again $p'$ lies outside $S'$.
\end{proof}

By a standard rank-space reduction we can assume that the corners of the hyperrectangles 
in $\S$ lie on the grid $[0,2m]^{d}$. After applying the above transformation, 
we note that the coordinates of the corners of the hypercubes in $\S'$ will 
lie on the grid $[-4m,0]^{d}\times [-2m,2m]^d$: 
trivially, $\Delta + a_i \leq 4m$, and hence $-4m \leq -a_i -\Delta 
\leq -a_i \leq 0$. Also, $-2m \leq b_i-\Delta \leq b_i\leq 2m$. 
After performing a suitable shifting of the grid, we will assume that all corners 
of the hypercubes in $\S'$ will lie on the grid $[0,4m]^{2d}$. 

\subsection{Constructing a bounded frequency set system.} 
We will now present a technique to select  a set $\hat{\S} \subseteq \S'$
with the following properties:
 \begin{enumerate}
 \item ({\em Bounded frequency}) Any point in $P'$ lies inside $O(\log^{2d}m)$ hypercubes in $\hat{\S}$.
 \item An $\alpha$-approximation dynamic set cover algorithm for $(P',\hat{\S})$ 
 implies an $O(\alpha\log^{2d-1}m)$-approximation dynamic set cover algorithm 
 for  $(P', {\cal S}')$.
 \item The time taken to update 
 the solution for  the set system $(P',\hat{\S})$ is
$O(\log^{2d}m\cdot\log^2n)$. 
   \item   The time taken to construct the set $\hat{\S}$ is $O(m\log^{2d}m)$.
 \end{enumerate}
 
 \subsubsection{Extended quad-tree for $2$-dimensional squares.}
 \label{subsec:ext_quad_2D}
 Given a set of squares $\S'$, construct a 2-dimensional quad-tree $\mT$ (as defined in Section~\ref{sec:Set-cover-squares}), such that its
 root contains all the squares in $\S'$.
We assume for simplicity that no two input squares in $\S'$ share a corner. Then, we can perturb the input points slightly so that no point $p\in P'$ lies on any of the grid points of the quad-tree and each square still contains the same set of points as before.
 Consider a node $v\in \mT$ and a square $S\in \S'$. 
 Let $C$ and $par(C)$ 
 be the cell corresponding to node $v$ and the parent node of $v$, respectively. 
 Let $\mathsf{proj}_i(C)$, $\mathsf{proj}_i(par(C))$ and $\mathsf{proj}_i(S)$ be the projection of $C$, $par(C)$ and $S$, respectively, on to the $i$-th dimension. Then $S$ is {\em $i$-long} at $v$ if and only if $\mathsf{proj}_i(C) \subseteq \mathsf{proj}_i(S)$ but 
 $\mathsf{proj}_i(par(C)) \not\subseteq \mathsf{proj}_i(S)$. See Figure~\ref{fig:long-assigned}(a).
 For all $u\in \mT$, let $\S(u,i) \subseteq \S'$ be 
 the squares which are $i$-long at
 node $u$. Intuitively, these are squares that cover the edge of $C$ in the $i$-th dimension but do not cover any edge of $par(C)$ in the $i$-th dimension. Now, at each node of $\mT$ we will construct 
 two {\em secondary structures} as follows: the first structure is a $1$-dimensional 
 quad-tree built on the projection of the squares in 
 $\S(u,1)$ on to the second dimension
 , and the second structure is a $1$-dimensional quad-tree built on the projection of the 
 squares in $\S(u,2)$ on to the first dimension.

 In each secondary structure, 
 an interval $I$ (corresponding to a square $S\in\S'$) is {\em assigned} to a node $u$ 
  if and only if $u$ is the node 
 with the smallest depth (the root is at depth zero) where $I$ intersects 
 either the left endpoint or the right endpoint of the cell $C_u$. See 
 Figure~\ref{fig:long-assigned}(b). By this definition, any interval will be assigned to 
 at most two nodes in the secondary structure.

Now we will use $\mT$ to construct the geometric collection $\hat{\S}$. 
 Let $V_{\text{sec}}$ be the set of nodes in all the secondary structures of $\mT$. 
 For any node $u\in V_{\text{sec}}$, among its assigned intervals which 
 intersect the left (resp., right) endpoint of the cell $C_u$,  
 identify the {\em maximal} interval 
 $I_{\ell}$ (resp., $I_r$) , i.e., the interval which has maximum overlap with $C_u$.
 See Figure~\ref{fig:long-assigned}(c). We then do the following set of operations over all the nodes in $V_{\text{sec}}$: For a node $u\in V_{\text{sec}}$, denote by $S'$ and $S''$ the corresponding squares for the assigned intervals $I_{\ell}$ and $I_r$, respectively. Further, let $w$ be the node in $\mT$, on which the secondary structure of $u$ was constructed. Then, we include in $\hat{\S}$ the rectangles $S_1\cap C_{w}$ and $S_2\cap C_{w}$. 

 \begin{figure}[h]
 \begin{center}
      \includegraphics[scale=1]{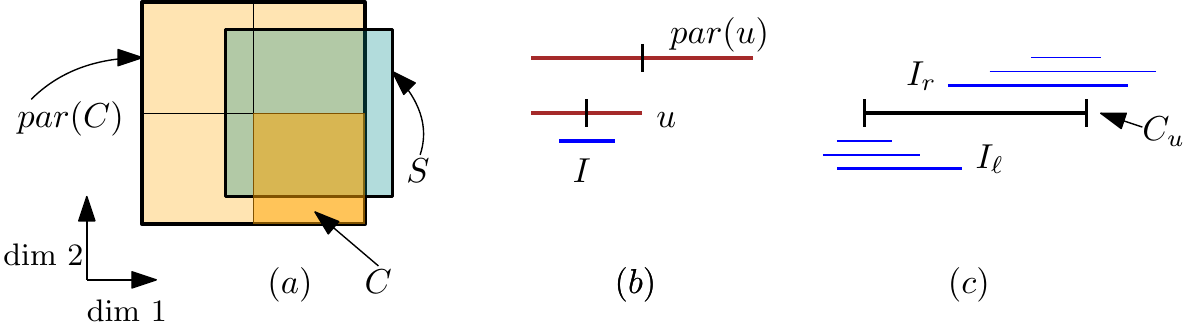}
 \end{center}

 \caption{(a) A square $S$ which is $1$-long at node $v$ (corrs.  
 cell $C$ is highlighted in darker orange), (b) $I$ is assigned to the two children of $v$, and (c) the maximal intervals 
 $I_{\ell}$ and $I_r$ at $C_v$.}
 \label{fig:long-assigned}
\end{figure}
  
 \subsubsection{Extended quad-tree for $2d$-dimensional hypercubes.} 
 
 In this section, we need a generalization of the quad-tree defined in \Cref{sec:Set-cover-squares}. For $d'>2$, a $d'$-dimensional quad-tree is defined
analogously to the the quad-tree defined in \Cref{sec:Set-cover-squares}, where instead of four, each internal node will now have $2^{d'}$ children.
 Assume by induction that we have defined how to construct the extended quad-tree for all dimensions less than or 
 equal to $2d{-}1$.
 (The base case is the extended quad-tree built for $2$-dimensional squares).
We define now how to construct the structure for $2d$-dimensional hypercubes. 
First construct the regular $2d$-dimensional quad-tree $\mT$ for the 
set of hypercubes $\S'$. 
 Consider any node $v\in \mT$. Generalizing the previous 
definition, for any $1\leq i \leq 2d$, a hypercube $S\in \S'$ is defined 
 to be {\em $i$-long} at node $v$ if and only if $\mathsf{proj}_i(C) \subseteq \mathsf{proj}_i(S)$, but 
 $\mathsf{proj}_i(par(C)) \not\subseteq \mathsf{proj}_i(S)$.  For all $v\in \mT$, let $\S(v,i) \subseteq \S'$ 
 be the hypercubes which are $i$-long at 
 node $v$.  Now, at each node of $\mT$ we will construct 
 $2d$ secondary structures as follows: for all $1\leq i \leq 2d$, the $i$-th 
 secondary structure is
 a $(2d{-}1)$-dimensional extended quad-tree built on $\S(v,i)$ and all its $2d$ dimensions 
 except the $i$-th dimension. Specifically, any hypercube $S \in \S(v,i)$ of the form 
 $\ell_1 \times \dots \times \ell_i \times \dots \times \ell_{2d}$ is projected to a $(2d{-}1)$-dimensional
 hypercube $\ell_1 \times\dots\times \ell_{i-1} \times \ell_{i+1}\times\dots \times \ell_{2d}$. 
Let  $\hat{\S}_v$ be the collection of the $(2d{-}1)$-dimensional hyperrectangles that are inductively picked for the secondary structure constructed at $v\in\mT$ using the routine $\mathsf{}$. Define the function $g$ which maps a $(2d{-}1)$-dimensional hyperrectangle picked as part of the collection $\hat{\S}_v$ (for a $v\in\mT$) to its corresponding $2d$-dimensional hypercube $S\in\S'$. We now define the collection of sets $\hat{\S}$ consisting of $2d$-dimensional hyperrectangles:
 \[\hat{\S} \leftarrow \bigcup_{v\in \mT}\left(\bigcup_{S'\in \hat{\S}_v} (g(S')\cap C_v)\right). \]

 

\begin{figure}[!ht]
\centering \includegraphics[scale=0.5]{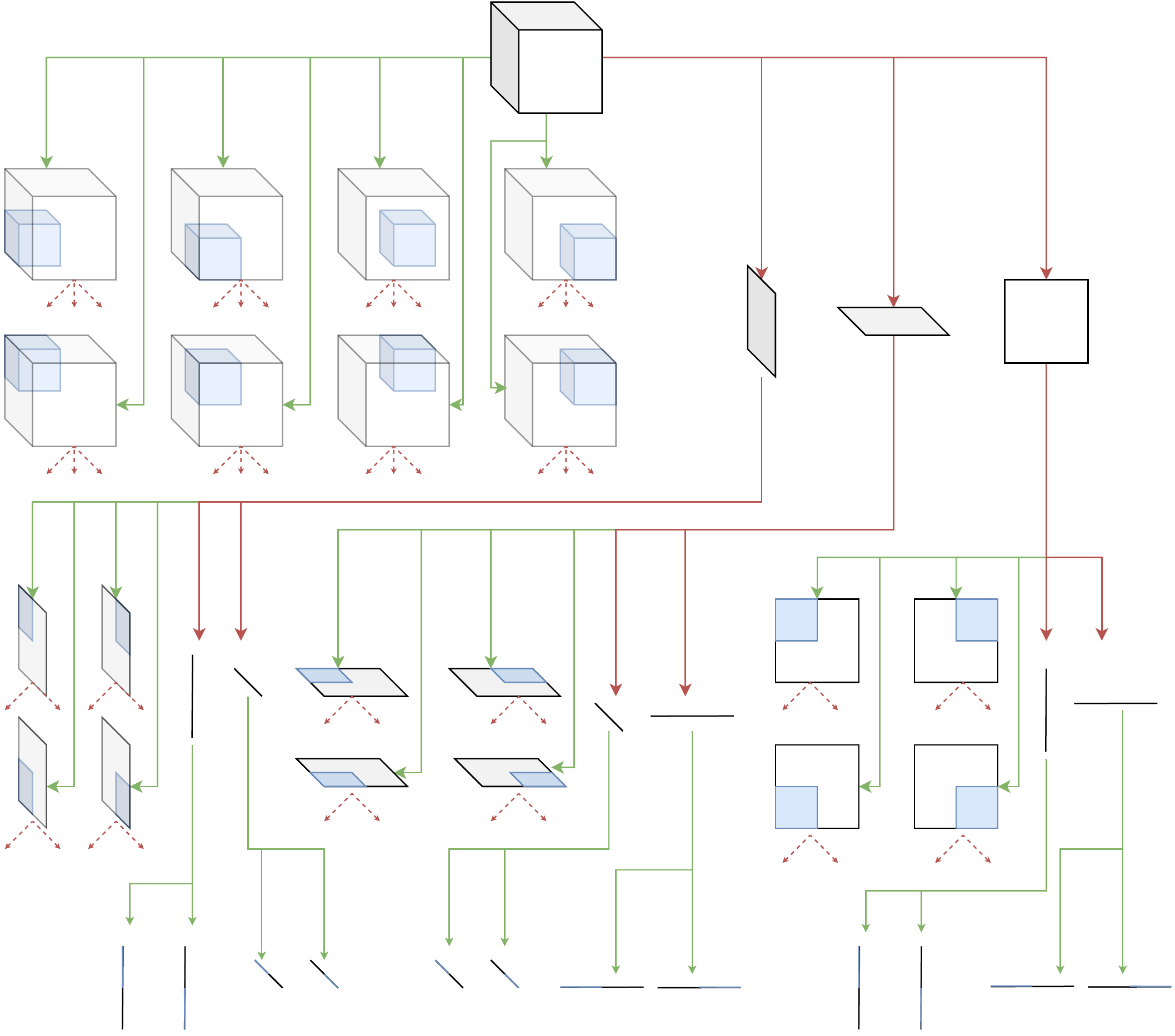}
\caption{Extended quad-tree with a $2\times2\times2$ cube as the root.}
\label{fig:extended-quadtree} 
\end{figure}


\begin{claim}
(Feasibility) Any point $p\in P'$ is covered by at least one set in the collection $\hat{\S}$.
 \end{claim}
\begin{proof}
We will prove the claim first for the case of 2-D squares. For any point $p\in P'$, we know that at least one square $S\in \S'$ covers it. By our assumption, $p$ is not on any of the grid points of the quad-tree. Then, we claim that $S$ is edge-covering for a cell $C_v$ such that $v$ is a leaf node and $C_v$ contains $p$. Then, by the definition of $i$-long, there exists an ancestor $u$ of $v$ in $\mT$ (or possibly $v$ itself) such that $S$ is $i$-long for $C_u$, for some $i\in [2]$. This implies that $\hat{\S}$ consists of a (maximal) square $S'$ such that $S'\in \S(u,i)$ and $S'\cap C_u\supseteq S\cap C_u$. Then, by our construction of the extended quad-tree, $S'\cup C_u$ is part of the collection $\hat{\S}$. Hence, the claim holds for 2-D squares. Generalizing this idea, feasibility can be guaranteed for the case of $2d$-dimensional hypercubes for $d\geq 1$.   
\end{proof}

 \begin{lemma}
 (Bounded frequency) Any point in $P'$ lies inside $O(\log^{2d}m)$ sets in $\hat{\S}$.
 \end{lemma}
 \begin{proof}
For the extended quad-tree for $2$-dimensional squares, let $\tau(2)$ be the maximum number of sets 
 in $\hat{\S}$ which contain a point $p=(p_x,p_y) \in P'$. By 
 properties of standard quad-tree, the number of nodes 
 in the $2$-dimensional quad-tree whose corresponding 
 cells contain $p$ is $O(\log m)$. At any such node, 
 each of the secondary structures will have $O(\log m)$ 
 nodes whose corresponding cells contain the projection 
 of $p$ (either $p_x$ or $p_y$). 
 At each cell in the secondary structure which contains $p$,
 we select at most two (maximal) hyperrectangles  into  $\hat{\S}$. Therefore, 
 $\tau(2)=O(\log^2m)$. 
 
 In general, for an extended quad-tree 
 for $d'$-dimensional hypercubes, let $\tau(d')$ be the maximum number of sets in $\hat{\S}$ that contain any given $d'$-dimensional point. Since we construct $2d$ 
 secondary structures at each node, we obtain the following 
 recurrence:
 
 \[\tau(2d) = O((2d)\log m) \cdot \tau(2d{-}1) =O(\log^{2d}m), \]     
 where the constant hidden by the big-$O$-notation depends on $d$. 
 \end{proof}

 
 \begin{lemma}\label{lem:alpha-approx}
 If there is an $\alpha$-approximation dynamic set cover algorithm for $(P',\hat{\S})$ 
 then there is an $O(\alpha\log^{2d-1}m)$-approximation dynamic set cover algorithm 
 for $(P', {\cal S}')$.
 \end{lemma}
 \begin{proof}
Let $\OPT$ and $\widehat{\OPT}$ be the size of the optimal set cover 
for $(P', \S')$ and $(P',\hat{\S})$. For any hypercube $S$ in $\S'$ we will 
show  that there exists 
$O(\log^{2d-1}m)$ hyperrectangles in $\hat{\S}$ whose union will completely 
cover $S$. Therefore, $\widehat{\OPT}= \OPT\cdot O(\log^{2d-1}m)$.
Note that for every set $S'\in\hat{\S}$, there exists a corresponding hypercube $S\in \S'$ which covers at least the set of points in $P'$ that $S'$ covers. 
For $(P', \S')$,
an $\alpha$-approximation set cover for $(P',\hat{\S})$
implies an approximation factor of:
\[\frac{b}{\OPT} = \frac{b\cdot O(\log^{2d-1}m)}{\widehat{\OPT}} = O(\alpha\log^{2d-1}m),\]
where the last equation follows from $b\leq \alpha \cdot\widehat{\OPT}$.

 Finally, we  establish the covering property.
 We will prove it via induction on the dimension size.
 As a base case, for squares in 2-D, let $\mu(2)$ be the number 
 of sets needed in $\hat{\S}$ such that 
 their union completely covers a square $S\in \S'$. 
 For any $S\in \S'$, 
  let $\mathsf{long}(S)$ be the set of nodes in 
$\mT$ where $S$ is $1$-long. By standard properties of a quad-tree, 
we have (a) $|\mathsf{long}(S)|=O(\log m)$, and (b) $S\leftarrow \bigcup_{v\in \mathsf{long}(S)} 
(S\cap C_v)$, where $C_v$ is the cell corresponding to $v$.
Now consider any node $v\in \mathsf{long}(S)$. Let $I$ be the interval corresponding 
to $S$ in the secondary structure of $v$ built on $\S(v,1)$. 
Via our selection of maximal intervals at the secondary nodes, it is clear that 
there exist two maximal intervals which cover $I$. 
Therefore, 
$\mu(2) \leq 2\cdot|\mathsf{long}(S)|=O(\log m)$. 
 
 In general, let $\mu(2d)$ be the number of sets needed in $\hat{\S}$ such that 
 their union completely covers a hypercube $S\in \S'$. Then we claim that 
 \[ \mu(2d) = O(2^d\log m)\times \mu(2d{-}1),\]
 where $O(2^d\log m)$ is the number of nodes in $\mT$ where $S$ is $1$-long.
Solving the recurrence, we obtain $\mu(2d) =O(\log^{2d-1}m)$.
\end{proof}

\subsection{The final algorithm}
For the general dynamic set cover problem, Bhattacharya {\em et al.}~\cite{bhattacharya2021dynamic} gave an $O(f)$-approximation 
algorithm with a worst-case update time of $O(f\log^2n)$. Recall 
that $f$ is the frequency of the set system. We will use their algorithm 
as a blackbox on the set system $(P',\hat{\S})$. Let $\ALG$ be the sets reported 
by their algorithm. Then our algorithm will also report $\ALG$ as a set cover for 
$(P',\S')$. The solution is feasible
since each set in $\hat{\S}$ belongs to $\S'$ as well. 

\begin{lemma}
The approximation factor of our algorithm is $O(\log^{4d-1}m)$.
\end{lemma}
\begin{proof}
For the set system $(P',\hat{\S})$ the frequency $f=O(\log^{2d}m)$, and hence,
 the algorithm of \cite{bhattacharya2021dynamic} leads to
an $O(\log^{2d}m)$ approximation for this set system. Using Lemma~\ref{lem:alpha-approx}, this 
implies an approximation factor of $O(\log^{2d}m) \cdot O(\log^{2d-1}m)=
O(\log^{4d-1}m)$ 
for the set system $(P',\S')$.
\end{proof}
\begin{lemma} 
 The update time  is $O(\log^{2d}m\cdot \log^2n)=O(\log^{2d+2}m)$.
 \end{lemma}
\begin{proof}
When a point is inserted or deleted, the $O(\log^{2d}m)$ 
sets in $\hat{\S}$ containing the point can be found in 
$O(\log^{2d}m)$ time by traversing the tree $\mT$. 
The algorithm of \cite{bhattacharya2021dynamic}
has an update time of $O(f\log^2n)=O(\log^{2d}m\cdot\log^2n)= O(\log^{2d+2}m)$. This is true since $d$-dimensional hypercubes have dual $VC$-dimension of $O(d)$ \cite{sauer1972density,shelah1972combinatorial} and hence, $O(\log n)=O(\log m)$.
\end{proof}

\begin{lemma}
  The time taken to construct the set $\hat{\S}$ is $O(m\log^{2d}m)$.
  \end{lemma}
  
  \begin{proof}
  Let $T(m,d)$ be the time taken to build the extended 
  quad-tree on $m$ hypercubes in $d$ dimensions.
  As a base case, we first compute $T(m,1)$. 
  Constructing the skeleton structure of the 
  1-dimensional quad-tree takes $O(m)$ time, 
  since the endpoints of the intervals lie on the integer grid $[0,4m]$.
  Then ``assigning'' each interval to a node in this 
  quad-tree takes $O(m\log m)$ time. For a node $v$ 
  which is assigned $m_v$ intervals, finding the two
  maximal intervals takes $O(m_v)$ time.
  Therefore, $T(m,1)= O(m\log m) + \sum_{v} O(m_v)=O(m\log m)$.
  
  Now consider the extended quad-tree on $m$ hypercubes in $2d$ dimensions.
  Again constructing the skeleton of the quad-tree takes only $O(m)$ time.
  For any $1\leq i\leq 2d$, finding the nodes in $\mT$ where a hypercube is 
  ``$i$-long'' takes $O(2^{2d}\log m)=O(\log m)$ time. Therefore,
  \begin{align*}
  T(m, 2d) &= O(m\log m) +  \sum_{v} T(m_v, 2d-1) \\
                &= O(m\log m) + \sum_{v} m_v\log^{2d-1}m_v \qquad\qquad \text{(by induction)}\\
                &= O(m\log m) + \sum_{v} m_v\log^{2d-1}m\\
                &= O(m\log^{2d}m), \qquad \qquad \text{since, $\sum_{v}m_v =O(m\log m)$.}
                \qedhere
  \end{align*}
 \end{proof} 

\subsection{Weighted setting}
We present an easy extension of our algorithm to the setting 
where each hyperrectangle $S \in \S$ has a weight 
$w_S \in [1, W]$. First, we {\em round} the weight of each 
set $S$ to the smallest power of two greater than or equal to
$w_S$. This leads to $O(\log W)$ different weight classes.
Next, for each weight class, 
we will build an extended quad-tree based on the hypercubes 
of that weight class after the reduction to the case of $2d$-dimensional hypercubes from $d$-dimensional hyperrectangles as shown in the previous section. Finally, let $\hat{\S}$ be the collection of 
(maximal) hypercubes obtained from all the $O(\log W)$ 
extended quad-trees. Run the dynamic set cover algorithm of 
Bhattacharya {\em et al.}~\cite{bhattacharya2021dynamic} 
on $(P',\hat{\S})$.

\begin{lemma}
The approximation factor of the algorithm is $O(\log^{4d-1}m\cdot \log W)$.
\end{lemma}

\begin{proof}
Consider the optimal solution for $(P',\S')$ and let $\OPT$ be 
the optimal weight. By rounding the 
weight of each set in the optimal solution, their total weight 
becomes at most $2\cdot \OPT$. Therefore, the weight of 
the optimal solution in the set system after rounding is 
at most $2\cdot \OPT$. Compared to the unweighted setting, 
now the frequency of the set system $(P',\hat{\S})$ increases 
by a $O(\log W)$ factor. As a result, we obtain an approximation 
factor of $O(\log^{4d-1}m\cdot \log W)$.
\end{proof}

\begin{lemma}
The update time of the algorithm is $O(\log^{2d}m\log^{3}(Wm))$. 
\end{lemma}
\begin{proof}
When a point is inserted or deleted, the $O(\log^{2d}m\log W)$ sets 
in $\hat{\S}$ containing the point can be found in 
$O(\log^{2d}m\log W)$ time by traversing the $O(\log W)$ extended 
quad-trees. The algorithm of \cite{bhattacharya2021dynamic} has an 
update time of $O(f\log^{2} Wn)=O(\log^{2d}m\log W\log^{2}Wn)=
O(\log^{2d}m\log^{3}(Wm))$. 
\end{proof}

\begin{theorem}\label{thm:WtDynSetCov}
There is an algorithm for weighted dynamic set cover for $d$-dimensional hyperrectangles 
with an approximation factor of $O(\log^{4d-1}m\cdot\log W)$
 and an update time 
of $O(\log^{2d}m\cdot\log^3(Wm))$. 
\end{theorem}
We also have the following corollary for dynamic set cover for $2d$-dimensional hypercubes where all the corners are integral and bounded in $[0,cm]^{2d}$ for some constant $c>0$.

\begin{corollary}\label{Cor:WtDynSetCov_cube}
There is an algorithm for weighted dynamic set cover for $2d$-dimensional hypercubes 
with an approximation factor of $O(\log^{4d-1}m\cdot\log W)$
 and an update time 
of $O(\log^{2d}m\cdot\log^3(Wm))$, when all of their corners are integral and bounded in $[0,cm]^{2d}$ for a fixed $c>0$. 
\end{corollary}

\section{Dynamic hitting set for $d$-dimensional hyperrectangles}
\label{sec:hit-set-hyperrectangles}
In this section we present a dynamic algorithm for hitting set for 
$d$-dimensional hyperrectangles. 
We will reduce the problem to an instance of dynamic set cover in $2d$-dimensional space and use the algorithm designed in the previous 
section (Theorem~\ref{Cor:WtDynSetCov_cube}).
Recall that $P$ is the set of 
points and $\S$ is the set of hyperrectangles.
Assume that all
the points of $P$ and the hyperrectangles in $\S$ lie in the box $[0, N]^d$. For all $p\in P$, we transform 
$p=(p_1,\ldots, p_d)$ to a $2d$-dimensional hypercube  $\S(p)$ 
of side length $N$ and 
``lower-left'' corner $p'=(-p_1,\ldots,-p_d, p_1,\ldots,p_d)$. 
Let $\S(P)$ be the  transformed hypercubes.
Next, we transform each hyperrectangle, say  $S \in \S$, with ``lower-left'' corner 
$a=(a_1,\ldots,a_d)$ and ``top-right'' corner $b=(b_1,\ldots,b_d)$ into a 
$2d$-dimensional point $P(S)=(-a_1,\ldots,-a_d,b_1,\ldots,b_d)$. 
Let $P(\S)$ be the  transformed points.
See Figure~\ref{fig:hit-duality}.

\begin{figure}[h]
 \begin{center}
      \includegraphics[scale=1]{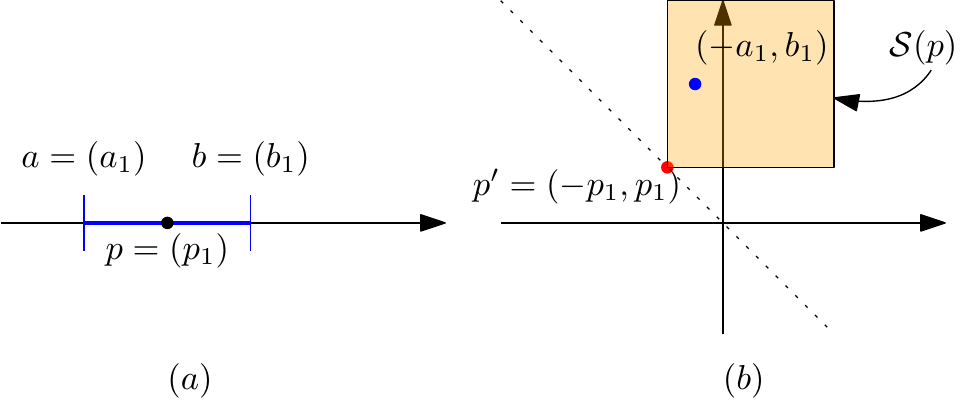}
 \end{center}
 \caption{(a) A point $p$ in 1-D lying inside an interval $S=[a_1,b_1]$, and 
 (b) the transformation of $p$ into a square $\S(p)$, and transformation of $S$ into a 
 point $(-a_1,b_1)$ in 2-D.}
 \label{fig:hit-duality}
\end{figure}

\begin{lemma} \label{lem:duality}
In $\IR^d$ a point $p$ lies inside a hyperrectangle $S$, if and only if, 
in $\IR^{2d}$ the hypercube $\S(p)$ contains the point $P(S)$.
\end{lemma}
\begin{proof}
 We define a point $q(q_1,\ldots,q_{2d})$ 
to {\em dominate} another point $q'(q'_1,\ldots,q'_{2d})$ if and only if 
$q_i \geq q'_i$, for all $1\leq i\leq 2d$. 
Assume that $p$ lies inside $S$. Then we have 
$p_i \geq a_i$ and $p_i \leq b_i$, for all $1\leq i\leq d$.
This implies that $P(S)=(-a_1,\ldots,-a_d,b_1,\ldots, b_d)$ dominates 
the point $p'(-p_1,\ldots,-p_d, p_1,\ldots,p_d)$, which is the lower-left 
corner of $\S(p)$.

The coordinates of the top-right corner of $\S(p)$ is 
$(-p_1+N,\ldots,-p_d+N, p_1+N,\ldots,p_d+N)$.
For all $1\leq i \leq d$, since $N\geq p_i-a_i$, it implies that 
$-p_i + N \geq -a_i$. For all $1\leq i \leq d$, since 
$p_i + N \geq N \geq b_i$, it implies that $ p_i + N \geq b_i$.
This finally implies that 
that the top-right corner of $\S(p)$ dominates $P(S)$. 
Therefore, $\S(p)$ contains $P(S)$.

Assume that $p$ lies outside $S$. Then there exists a 
dimension $i$ such that $p_i < a_i$ or $p_i > b_i$. 
If $p_i < a_i$, then $-p_i > -a_i$, which implies that 
$P(S)$ cannot dominate the lower-left corner of $\S(p)$.
If $p_i > b_i$, again $P(S)$ cannot dominate the 
lower-left corner of $\S(p)$. This implies that 
$P(S)$ lies outside $\S(p)$.
\end{proof}

We will use the above reduction 
to transform the points in $P$
into  hypercubes in $\IR^{2d}$ and transform the hyperrectangles 
in $\S$ into points in $\IR^{2d}$. 
Therefore, the hitting set problem in 
$\IR^d$ on hyperrectangles has been 
reduced to the set cover problem in $\IR^{2d}$
on hypercubes. And with a similar rank-space reduction as mentioned in the previous section, all the points as well as the corners of the hypercubes in this instance have integral coordinates.
Then, the set cover instance is answered using ~\Cref{Cor:WtDynSetCov_cube}. The correctness follows from \Cref{lem:duality}. Noting that for $2d$-dimensional hypercubes, $O(\log n)=O(\log m)$, the performance of the algorithm is summarized 
below.

\begin{theorem}\label{thm:WtDynHitSet}
After performing a pre-processing step which takes 
$O(n\log^{2d}n)$ time, there is an algorithm for 
hitting set for $d$-dimensional hyperrectangles 
with an approximation factor of $O(\log^{4d-1}n)$ and 
an update time of $O(\log^{2d+2}n)$.
In the weighted setting, the approximation factor 
is $O(\log^{4d-1}n\cdot\log W)$ 
and the update time is $O(\log^{2d}n\log^3(Wn))$.
\end{theorem}

\section{Future work}

In the first part of this work, we have studied online geometric set cover 
and hitting set for 2-D squares. This opens up an interesting line of work 
for the future. We state a few open problems: 
\begin{enumerate}
    \item As a natural extension of 2-D squares, is it possible to 
    design a $o(\log n \log m)$-competitive algorithm for 3-D cubes? 
    The techniques used in this paper for 2-D squares do not seem to 
    extend to 3-D cubes. Another setting of interest here is when the geometric objects are 2-D disks. Can we obtain a $o(\log n \log m)$-competitive online set cover algorithm for them?
    \item Design an online algorithm for set cover (resp., hitting set) for rectangles with competitive ratio $o(\log^2 n)$ or show an almost matching lower bound of $\Omega\left(\frac{\log^2 n}{\log\log n}\right)$ (which holds for the general case of online set cover~\cite{alon2003online})?
    \item As a generalization of the above question, is it possible to obtain online algorithms for set cover and hitting set with competitive ratio $o(\log^2 n)$ for set systems with ``constant'' VC-dimension.
    \item For the weighted case of online set cover, even in unit squares, can we obtain algorithms with competitive ratio $o(\log n\log m)$?
    \item Design an online algorithm for hitting set for squares with competitive ratio $O(\log n)$, and hence, improving our algorithm's competitive ratio of $O(\log N)$ (where the corners of the squares are integral and contained in $[0,N)^2$).
\end{enumerate}

In the second part of our work, we studied dynamic geometric set cover 
and hitting set for $d$-dimensional hyperrectangles. This line of work 
nicely brings together data structures, computational geometry, and approximation 
algorithms. We finish with a few open problems in the dynamic setting:
\begin{enumerate}
    \item Improve the approximation factor for dynamic set cover for the case of  2-D rectangles. Specifically, is it possible to obtain an $O(\log n)$ approximation with polylogarithmic update time? In this setting the 
    rectangles are fixed, but the points are dynamic.
    \item For weighted dynamic set cover for the case of 2-D rectangles, is it possible to obtain approximation and update bounds independent of $W$ (where $W$ is the ratio of the weight of the highest weight rectangle to the lowest weight rectangle in the input)?
    \item For the (fully) dynamic case of set cover studied in~\cite{agarwal2020dynamic, chan2021more, chan2022dynamic}, can we obtain algorithms with sublinear update time and polylogarithmic approximation when the sets are rectangles (as {originally asked in~\cite{chan2022dynamic}})?
\end{enumerate}

\bibliographystyle{plainurl}
\bibliography{ref_dyn}

\begin{thebibliography}{10}

\bibitem{abboud2019dynamic}
Amir Abboud, Raghavendra Addanki, Fabrizio Grandoni, Debmalya Panigrahi, and
  Barna Saha.
\newblock Dynamic set cover: improved algorithms and lower bounds.
\newblock In {\em Proceedings of the 51st Annual ACM SIGACT Symposium on Theory
  of Computing (STOC)}, pages 114--125, 2019.

\bibitem{AgarwalCSXX20}
Pankaj~K. Agarwal, Hsien{-}Chih Chang, Subhash Suri, Allen Xiao, and Jie Xue.
\newblock Dynamic geometric set cover and hitting set.
\newblock In {\em 36th International Symposium on Computational Geometry, SoCG
  2020}, volume 164 of {\em LIPIcs}, pages 2:1--2:15. Schloss Dagstuhl -
  Leibniz-Zentrum f{\"{u}}r Informatik, 2020.

\bibitem{agarwal2020dynamic}
Pankaj~K Agarwal, Hsien-Chih Chang, Subhash Suri, Allen Xiao, and Jie Xue.
\newblock Dynamic geometric set cover and hitting set.
\newblock {\em arXiv preprint arXiv:2003.00202}, 2020.

\bibitem{alon2003online}
Noga Alon, Baruch Awerbuch, and Yossi Azar.
\newblock The online set cover problem.
\newblock In {\em Proceedings of the thirty-fifth annual ACM symposium on
  Theory of computing (STOC)}, pages 100--105, 2003.

\bibitem{aronov2010small}
Boris Aronov, Esther Ezra, and Micha Sharir.
\newblock Small-size $\backslash$eps-nets for axis-parallel rectangles and
  boxes.
\newblock {\em SIAM Journal on Computing}, 39(7):3248--3282, 2010.

\bibitem{arora1998polynomial}
Sanjeev Arora.
\newblock Polynomial time approximation schemes for euclidean traveling
  salesman and other geometric problems.
\newblock {\em Journal of the ACM (JACM)}, 45(5):753--782, 1998.

\bibitem{arya1998optimal}
Sunil Arya, David~M. Mount, Nathan~S. Netanyahu, Ruth Silverman, and Angela~Y.
  Wu.
\newblock An optimal algorithm for approximate nearest neighbor searching fixed
  dimensions.
\newblock {\em Journal of the ACM (JACM)}, 45(6):891--923, 1998.

\bibitem{assadi2021fully}
Sepehr Assadi and Shay Solomon.
\newblock Fully dynamic set cover via hypergraph maximal matching: An optimal
  approximation through a local approach.
\newblock In {\em 29th Annual European Symposium on Algorithms, ESA 2021},
  page~8. Schloss Dagstuhl-Leibniz-Zentrum fur Informatik GmbH, Dagstuhl
  Publishing, 2021.

\bibitem{berg1997computational}
Mark~de Berg, Marc~van Kreveld, Mark Overmars, and Otfried Schwarzkopf.
\newblock Computational geometry.
\newblock In {\em Computational geometry}, pages 1--17. Springer, 1997.

\bibitem{bhattacharya2018deterministic}
Sayan Bhattacharya, Monika Henzinger, and Giuseppe~F. Italiano.
\newblock Deterministic fully dynamic data structures for vertex cover and
  matching.
\newblock {\em SIAM Journal on Computing}, 47(3):859--887, 2018.

\bibitem{bhattacharya2018dynamic}
Sayan Bhattacharya, Monika Henzinger, and Giuseppe~F. Italiano.
\newblock Dynamic algorithms via the primal-dual method.
\newblock {\em Information and Computation}, 261:219--239, 2018.

\bibitem{bhattacharya2019new}
Sayan Bhattacharya, Monika Henzinger, and Danupon Nanongkai.
\newblock A new deterministic algorithm for dynamic set cover.
\newblock In {\em 2019 IEEE 60th Annual Symposium on Foundations of Computer
  Science (FOCS)}, pages 406--423. IEEE, 2019.

\bibitem{bhattacharya2021dynamic}
Sayan Bhattacharya, Monika Henzinger, Danupon Nanongkai, and Xiaowei Wu.
\newblock Dynamic set cover: Improved amortized and worst-case update time.
\newblock In {\em Proceedings of the 2021 ACM-SIAM Symposium on Discrete
  Algorithms (SODA)}, pages 2537--2549. SIAM, 2021.

\bibitem{bhore2020dynamic}
Sujoy Bhore, Jean Cardinal, John Iacono, and Grigorios Koumoutsos.
\newblock Dynamic geometric independent set.
\newblock In {\em Japan conference on Discrete and Computational Geometry,
  Graphs, and Games}, 2021.

\bibitem{bienkowski2020unbounded}
Marcin Bienkowski, Jaros{\l}aw Byrka, Christian Coester, and {\L}ukasz Je{\.z}.
\newblock Unbounded lower bound for k-server against weak adversaries.
\newblock In {\em Proceedings of the 52nd Annual ACM SIGACT Symposium on Theory
  of Computing (STOC)}, pages 1165--1169, 2020.

\bibitem{BuchbinderN09}
Niv Buchbinder and Joseph Naor.
\newblock The design of competitive online algorithms via a primal-dual
  approach.
\newblock {\em Found. Trends Theor. Comput. Sci.}, 3(2-3):93--263, 2009.

\bibitem{cardinal2021worst}
Jean Cardinal, John Iacono, and Grigorios Koumoutsos.
\newblock Worst-case efficient dynamic geometric independent set.
\newblock In {\em 29th Annual European Symposium on Algorithms (ESA 2021)}.
  Schloss Dagstuhl-Leibniz-Zentrum f{\"u}r Informatik, 2021.

\bibitem{chan2012weighted}
Timothy~M. Chan, Elyot Grant, Jochen K{\"o}nemann, and Malcolm Sharpe.
\newblock Weighted capacitated, priority, and geometric set cover via improved
  quasi-uniform sampling.
\newblock In {\em Proceedings of the twenty-third annual ACM-SIAM symposium on
  Discrete Algorithms (SODA)}, pages 1576--1585. SIAM, 2012.

\bibitem{chan2021more}
Timothy~M. Chan and Qizheng He.
\newblock More dynamic data structures for geometric set cover with sublinear
  update time.
\newblock In {\em 37th International Symposium on Computational Geometry (SoCG
  2021)}. Schloss Dagstuhl-Leibniz-Zentrum f{\"u}r Informatik, 2021.

\bibitem{chan2022dynamic}
Timothy~M. Chan, Qizheng He, Subhash Suri, and Jie Xue.
\newblock Dynamic geometric set cover, revisited.
\newblock In {\em Proceedings of the 2022 Annual ACM-SIAM Symposium on Discrete
  Algorithms (SODA)}, pages 3496--3528. SIAM, 2022.

\bibitem{clarkson2005improved}
Kenneth~L. Clarkson and Kasturi Varadarajan.
\newblock Improved approximation algorithms for geometric set cover.
\newblock In {\em Proceedings of the twenty-first annual symposium on
  Computational geometry (SoCG)}, pages 135--141, 2005.

\bibitem{dallant2021conditional}
Justin Dallant and John Iacono.
\newblock Conditional lower bounds for dynamic geometric measure problems.
\newblock {\em arXiv preprint arXiv:2112.10095}, 2021.

\bibitem{EvenS11}
Guy Even and Shakhar Smorodinsky.
\newblock Hitting sets online and vertex ranking.
\newblock In {\em Algorithms - {ESA} 2011 - 19th Annual European Symposium},
  volume 6942 of {\em Lecture Notes in Computer Science}, pages 347--357.
  Springer, 2011.

\bibitem{gupta2022random}
Anupam Gupta, Gregory Kehne, and Roie Levin.
\newblock Random order online set cover is as easy as offline.
\newblock In {\em 2021 IEEE 62nd Annual Symposium on Foundations of Computer
  Science (FOCS)}, pages 1253--1264. IEEE, 2022.

\bibitem{gupta2017online}
Anupam Gupta, Ravishankar Krishnaswamy, Amit Kumar, and Debmalya Panigrahi.
\newblock Online and dynamic algorithms for set cover.
\newblock In {\em Proceedings of the 49th Annual ACM SIGACT Symposium on Theory
  of Computing (STOC)}, pages 537--550, 2017.

\bibitem{gupta2020fully}
Anupam Gupta and Roie Levin.
\newblock Fully-dynamic submodular cover with bounded recourse.
\newblock In {\em 2020 IEEE 61st Annual Symposium on Foundations of Computer
  Science (FOCS)}, pages 1147--1157. IEEE, 2020.

\bibitem{gupta2020online}
Anupam Gupta and Roie Levin.
\newblock The online submodular cover problem.
\newblock In {\em Proceedings of the Fourteenth Annual ACM-SIAM Symposium on
  Discrete Algorithms (SODA)}, pages 1525--1537. SIAM, 2020.

\bibitem{har2012weighted}
Sariel Har-Peled and Mira Lee.
\newblock Weighted geometric set cover problems revisited.
\newblock {\em Journal of Computational Geometry}, 3(1):65--85, 2012.

\bibitem{Henzinger0W20}
Monika Henzinger, Stefan Neumann, and Andreas Wiese.
\newblock Dynamic approximate maximum independent set of intervals, hypercubes
  and hyperrectangles.
\newblock In {\em 36th International Symposium on Computational Geometry, SoCG
  2020}, volume 164 of {\em LIPIcs}, pages 51:1--51:14. Schloss Dagstuhl -
  Leibniz-Zentrum f{\"{u}}r Informatik, 2020.

\bibitem{mustafa2014settling}
Nabil~H. Mustafa, Rajiv Raman, and Saurabh Ray.
\newblock Settling the apx-hardness status for geometric set cover.
\newblock In {\em 2014 IEEE 55th Annual Symposium on Foundations of Computer
  Science (FOCS)}, pages 541--550. IEEE, 2014.

\bibitem{MustafaR09}
Nabil~H. Mustafa and Saurabh Ray.
\newblock {PTAS} for geometric hitting set problems via local search.
\newblock In {\em Proceedings of the 25th {ACM} Symposium on Computational
  Geometry (SoCG)}, pages 17--22. {ACM}, 2009.

\bibitem{sauer1972density}
Norbert Sauer.
\newblock On the density of families of sets.
\newblock {\em Journal of Combinatorial Theory, Series A}, 13(1):145--147,
  1972.

\bibitem{shelah1972combinatorial}
Saharon Shelah.
\newblock A combinatorial problem; stability and order for models and theories
  in infinitary languages.
\newblock {\em Pacific Journal of Mathematics}, 41(1):247--261, 1972.

\bibitem{varadarajan2010weighted}
Kasturi Varadarajan.
\newblock Weighted geometric set cover via quasi-uniform sampling.
\newblock In {\em Proceedings of the forty-second ACM symposium on Theory of
  computing (STOC)}, pages 641--648, 2010.

\end{thebibliography}

\appendix

\section{Online algorithms for interval set cover}
\label{sec:interval_upper}
In this section, we present a tight 2-competitive algorithm for the case of interval set cover.

In the algorithm, we start with an empty set cover. 
In each iteration, when a new point $p$ arrives, if it is covered then we do nothing.
Otherwise, we select among the intervals covering  $p$, the one with the rightmost right end-point and the one with the leftmost left end-point.

The correctness of the algorithm follows trivially, since for every new uncovered point we pick an
interval covering it. We do not remove intervals from our solution at any later steps in the algorithm, and
hence, all points are covered when the algorithm terminates.

\begin{restatable}{theorem}{onlineintervalupper}
\label{thm:onlineintervalupper}
    There exists a 2-competitive algorithm for the  online interval set cover problem.
\end{restatable}
\begin{proof}
    Consider an interval $I$ in the optimum solution $\OPT$. When the first uncovered point covered by
    it, arrives in the input, our algorithm picks two intervals and  ensures that these two intervals cover all of  $I$. Hence, for each interval in $\OPT$,
    we pick at most 2 intervals in our solution, giving us a 2-competitive solution.
\end{proof}

\begin{restatable}{theorem}{onlineintervallower}
\label{thm:onlineintervallower}
    There is an instance of the set cover problem on intervals such that any online algorithm
    (without recourse) can at best be 2-competitive on this instance.
\end{restatable}
\begin{proof}
    Consider the given set of intervals to be $A:=[0,1],B:=[1,2],C:=[2,3],D:=[3,4]$. The first point to
    arrive is $p_1:=2$. 
    If the algorithm picks two or more sets, then we are done as $\OPT$ is of size 1. 
    Otherwise, to cover $p$, an algorithm can pick either interval  $B$ or  $C$.  In the
    former case, the second point should be  $p_2:=3$; and in the latter  $p_3:=1$.  We see that in both cases $\OPT$ is of size one, but an online algorithm is forced to pick two intervals.
\end{proof}

\end{document}